\documentclass[a4paper,onecolumn,11pt, accepted=2024-04-08]{quantumarticle}
\pdfoutput=1

\usepackage[utf8]{inputenc}
\usepackage[T1]{fontenc} 
\usepackage{braket}

\usepackage{amsfonts, amssymb, amsmath, amsthm}
\usepackage{mathtools}
\usepackage[american]{babel}

\usepackage{graphicx}

\usepackage[colorlinks=true,citecolor=blue,linkcolor=blue,urlcolor=red]{hyperref}
\usepackage{xurl}
\hypersetup{breaklinks=true}

\usepackage[numbers,sort&compress]{natbib}
\usepackage{chngcntr}

\renewcommand\vec{\mathbf}
\newcommand*{\p}{\vec{p}}
\newcommand*{\q}{\vec{q}}
\newcommand*{\e}{\vec{e}}
\newcommand*{\s}{\vec{s}}

\DeclareMathOperator{\locc}{LOCC}

\DeclareMathOperator{\pure}{PURE}
\DeclareMathOperator{\prob}{Prob}
\DeclareMathOperator{\SR}{SR}
\DeclareMathOperator{\diag}{diag}

\DeclarePairedDelimiter\floor{\lfloor}{\rfloor}
\DeclarePairedDelimiter\ceil{\lceil}{\rceil}

\newcommand*{\N}[1]{\left\lVert#1\right\rVert}
\newcommand*{\A}[1]{\left\lvert#1\right\rvert}
\newcommand*{\KBra}[2]{|#1\rangle\!\langle#2|}

\newcommand*{\mb}{\mathfrak{B}}
\newcommand*{\md}{\mathfrak{D}}

\newcommand*{\mM}{\mathcal{M}}
\newcommand*{\mN}{\mathcal{N}}
\newcommand*{\mbN}{\mathbb{N}}
\newcommand*{\mbR}{\mathbb{R}}

\newcounter{theorems}
\newtheorem{theorem}[theorems]{Theorem}
\newtheorem{definition}[theorems]{Definition}
\newtheorem{lemma}[theorems]{Lemma}
\newtheorem{proposition}[theorems]{Proposition}
\newtheorem{corollary}[theorems]{Corollary}

\makeatletter
\providecommand*{\diff}%
{\@ifnextchar^{\DIfF}{\DIfF^{}}} \def\DIfF^#1{%
	\mathop{\mathrm{\mathstrut d}}% 
	\nolimits^{#1}\gobblespace}
\def\gobblespace{% 
	\futurelet\diffarg\opspace}
\def\opspace{%
	\let\DiffSpace\!%
	\ifx\diffarg(%
	\let\DiffSpace\relax
	\else
	\ifx\diffarg[%
	\let\DiffSpace\relax
	\else
	\ifx\diffarg\{%
	\let\DiffSpace\relax
	\fi\fi\fi\DiffSpace}
\providecommand*{\deriv}[3][]{% 
	\frac{\diff^{#1}#2}{\diff #3^{#1}}}

\begin{document}

\title{Complete Characterization of Entanglement Embezzlement}
\author{Elia Zanoni}
\email{elia.zanoni@ucalgary.ca}
\orcid{0009-0002-7494-2968}
\author{Thomas Theurer}
\orcid{0000-0002-3888-4628}
\author{Gilad Gour}
\orcid{0000-0002-4892-4072}
\affiliation{Department of Mathematics and Statistics, University of Calgary, Calgary, AB T2N 1N4, Canada}
\affiliation{Institute for Quantum Science and Technology, University of Calgary, Calgary, AB T2N 1N4, Canada}

\begin{abstract}
	Using local operations and classical communication (LOCC), entanglement can be manipulated but not created. However, entanglement can be \emph{embezzled}.
    In this work, we completely characterize universal embezzling families and demonstrate how this singles out the original family introduced by van Dam and Hayden. To achieve this, we first give a full characterization of pure to mixed state LOCC-conversions. Then, we introduce a new conversion distance and derive a closed-form expression for it. 
    These results might be of independent interest.
\end{abstract}

\maketitle

\section{Introduction}
Quantum entanglement~\cite{BBPS96, VPRK97, Nie99, BPRST00, PV07, HHHH09} describes correlations between different particles with no classical counterpart and has both deep foundational implications~\cite{EPR35} and numerous applications in quantum information science~\cite{BB84a, BB84b, BWS92, BBCJPW93, BdVSW96, BSST99, BSST06, BDHSW14, Pop94, HHH99, BCR11}.
If two or more parties are far apart, in practice, they are restricted to local operations and classical communication (LOCC) because during transmission over long distances, physical systems will unavoidably interact with an environment and eventually lose the quantum information they carry. In contrast, it is much simpler to exchange classical information that can easily be amplified and protected. The states that can be prepared with $\locc$, i.e., the separable ones, are therefore considered free and the ones that cannot, which are exactly the entangled ones, are considered costly or resourceful.

According to this reasoning, entanglement is a resource and thus studied within the framework of quantum resource theories (QRTs) \cite{CG19, CFS16}. In a QRT, a physically motivated restriction such as the one discussed above divides both states and operations into free or resourceful in a consistent manner: Free operations map free states into free states. Once the free operations and states have been fixed, a QRT studies which quantum advantages depend on the resource under consideration and, closely related, how the consumption of resourceful states can help to overcome the restriction. An example is quantum teleportation~\cite{BBCJPW93}: By consuming entangled states, $\locc$ allows to teleport quantum systems and thereby simulate arbitrary operations outside $\locc$.

In this sense, consuming entangled states or more generally resourceful states can lead to operational advantages, e.g., in communication scenarios~\cite{BB84a, BB84b, BWS92, BBCJPW93, BdVSW96, BSST99, BSST06, BDHSW14}. Importantly, if a state can be converted with a free operation into another one, then the former is at least as valuable as the latter in any application that only allows for free operations. Answering the question which states can be converted into each other is thus a central question in any QRT. In entanglement theory, Nielsen's Theorem~\cite{Nie99} provides a characterization of the deterministic conversion between pure states that was later extended to probabilistic conversions~\cite{JP99A}. Here, we generalize these results and characterize the conversion of a pure state into a mixed state. 

The results on state conversions mentioned so far require that an initial state is exactly converted into a target state. However, every physical setup is affected by noise and finitely many measurements will only lead to finite accuracy. It is therefore practically impossible to distinguish the desired target state from a close approximation. This is the physical motivation to investigate whether a state can be \textit{approximately} converted to the target using free operations~\cite{VJN00}. There are several equivalent distances that formalize a notion of closeness between quantum states \cite{Uhl76, Joz94, FVdG99, TCR10, GT20}, which can be used to define conversion distances~\cite{WW19, SCG20, SG22, Gou22}. These conversion distances measure the error in approximate conversions, i.e., given two states, they are the smallest distance between the second state and the result of any free operation on the first one. In this article, we introduce a new conversion distance under $\locc$ defined on pure states. Using our result on exact conversions, we prove that it is topologically equivalent to the other conversion distances and derive a closed formula for it.

This in turn will allow us to completely characterize entanglement embezzlement~\cite{vDH03, LW14}: Whilst it is impossible to create entanglement with LOCC~\cite{VPRK97, PV07, CLP17, CG19}, it is possible to embezzle it in the sense that one converts a given entangled state approximately to itself \textit{and} a copy of another entangled target state. A family of pure states such that one can do this for any target state with arbitrary accuracy using $\locc$ and one of its members is called a universal embezzling family. Embezzlement is thus a generalization of quantum catalysis, a phenomenon discovered in the early years of the resource theory of entanglement~\cite{JP99O}: Whilst in catalysis, the catalyst must be preserved exactly, embezzlement allows to change it by an arbitrarily small amount. Very recently, we have witnessed a renewed interest in the catalysis of various quantum resources~\cite{LRS23, GKS23, SN23, LN23, LWW23, DGKS23, SN22, LPB23, KL22, DKMS22, LS21, KDS21, Kar21, Wil21, LJ21, LS21a, SS21}, with applications far beyond state conversion. In Ref.~\cite{LS21}, it was for example shown that a catalytic quantum teleportation protocol outperforms the standard teleportation protocol~\cite{BBCJPW93, Pop94, HHH99} in terms of teleportation fidelity. For further details and examples, see the recent review articles Refs.~\cite{LW23, DKMS22a}.

Universal embezzling families are valuable resources in various applications. They are, for example, necessary resources for the efficient simulation of noisy quantum channels with noiseless channels~\cite{BDHSW14, BCR11}, a result known in quantum information theory as `Quantum Reverse Shannon Theorem'. They are also an important component in the elementary proofs of Grothendiek theorems~\cite{RV14, Gro53}, which are of fundamental importance in the theories of Banach spaces and $C^*$- algebras (see, e.g., Ref~\cite{Pis12}). Furthermore, embezzling families are necessary to win various quantum guessing games with certainty~\cite{LTW13, DSV15, RV15, JLV20}. Lastly, due to the close relation to catalysis, we expect that they will prove useful in applications such as teleportation, where catalysis provides advantages.

The first universal embezzling family was introduced in Ref.~\cite{vDH03} and we call it van Dam and Hayden embezzling family. More recently, additional families have been proposed in Ref.~\cite{LW14}. In these works, embezzlement is considered using only local operations. Here, we extend the framework and additionally allow for classical communication, provide a complete characterization of universal embezzling families under LOCC, and discuss in what sense the van Dam and Hayden family is unique.

\section{Notation and Preliminaries}
In this article, we restrict ourselves to finite-dimensional Hilbert spaces and denote them with capital Latin letters such as $A$, $B$. The dimension of a Hilbert space $C$ is denoted by $\A{C}$, and the set of density matrices acting on it by $\md{(C)}$. For the set of pure states in $\md(C)$, we write $\pure{(C)}$.
Small Greek letters such as $\rho$ and $\sigma$ denote density matrices, with $\psi$, $\varphi$, and $\chi$ reserved for pure states. For $\rho, \sigma \in \md(C)$, the trace distance between $\rho$ and $\sigma$ is defined as $\frac{1}{2}\N{\rho - \sigma}_1$, where $\N{\cdot}_1$ is the trace norm.

Probability vectors are represented by bold small Latin letters, e.g., $\p$, with $p_x$ the $x$-th component of $\p$.   The set of probability vectors of length $d$ is denoted by $\prob(d)$, which contains the subset $\prob^\downarrow(d)$ consisting of all $d$-dimensional probability vectors with non-increasing entries. The $k$-th Ky Fan norm of $\p \in \prob^\downarrow (d)$ is defined as
\begin{equation}
	\N{\p}_{(k)} = \sum_{x=1}^k p_x
\end{equation}
for $k \in[d]$, where $[d]$ is a shorthand notation for $\{1, \dots, d\}$. We write $\p \succ \q$ if  $\p$ majorizes $\q$~\cite{MOA11, NV01, Nie02} and $\N{\p-\q}_1$ for $\sum_x |p_x-q_x|$.

Since this article is concerned with two spatially separated parties, call them Alice and Bob, it is important to make clear which system is under the control of whom: Systems belonging to Alice will always be denoted by $A$ or $A'$, and systems belonging to Bob by $B$ or $B'$.
Quantum channels are represented by calligraphic large Latin letters such as $\mM$, $\mN$, and the set of quantum channels from a bipartite system $AB$ to $A'B'$ that Alice and Bob can implement if they are restricted to local operations and classical communication is denoted by $\locc(AB \to A'B')$.
For bipartite systems $AB$, we assume w.l.o.g. that $\A{A} = \A{B}=d$ (since with $\locc$, Alice and Bob can always attach and remove local auxiliary systems). Fixing an orthonormal basis $\Set{\Ket{x}_A}_{x \in [d]}$ for $A$ and  $\Set{\Ket{x}_B}_{x \in [d]}$ for $B$, under $\locc$, every $\psi \in \pure(AB)$ is then equivalent to its standard form $\sum_x \sqrt{p_x} \Ket{xx}_{AB}$, where $\p \in \prob^\downarrow(d)$ are the Schmidt coefficients of $\psi$ (which, w.l.o.g., we will always assume to be ordered non-increasingly from here on).     We denote with $\SR(\psi)$ the Schmidt rank of a pure bipartite state $\psi\in \pure(AB)$, i.e., the number of non-zero Schmidt coefficients of $\psi$, and use the symbol $\Phi_m = \sum_{x = 1}^m \frac{1}{\sqrt{m}}\Ket{xx}_{AB}$ for the maximally entangled state on $AB$, where $\A{A} = \A{B} = m$. Finally, for $\rho \in \md(AB)$ and $\sigma \in \md(A'B')$,  we write $\rho \xrightarrow\locc \sigma$ whenever there exists an $\mN \in \locc(AB \to A'B')$ such that $\sigma = \mN(\rho)$ and $\rho \xrightarrow{\locc} \Set{t_z, \tau_z}$ whenever there exists a probabilistic $\locc$ protocol that converts $\rho$ to $\tau_z$ with probability $t_z$.

\section{State conversions with LOCC}
As motivated in the introduction,  the question of how entangled states can be interconverted is at the core of the resource theory of entanglement. The exact deterministic $\locc$-conversion problem between pure states is solved by Nielsen's Theorem~\cite{Nie99}: For $\psi$, $\varphi \in \pure(AB)$ and $\p$, $\q \in \prob^\downarrow{(|A|)}$ their associated Schmidt vectors, $\psi \xrightarrow{\locc} \varphi$ if and only if~\cite{Vid99}
\begin{equation}
	E_{k}(\psi)  \geq E_{k}(\varphi)   \quad  	\,\forall k \in [|A|]\;,
\end{equation}
where $E_{k}(\psi) := 1 - \N{\p}_{(k)}$.
In Ref.~\cite{JP99A}, this result was generalized to the case where the target state $\varphi$ is replaced with an ensemble of states: For $\psi$, $\varphi_1$, $\dots$, $\varphi_n \in \pure(AB)$, $\psi \xrightarrow{\locc} \Set{t_z, \varphi_z}_{z \in[n]}$ if and only if
\begin{equation}
	E_{k}(\psi)  \geq \sum_{z=1}^nt_zE_{k}(\varphi_z) \quad \forall k \in [|A|],
\end{equation}
which can be rewritten as
\begin{equation}\label{eq:pure-to-prob-mixture}
	\min_{k\in [|A|]}\Bigl\{E_{k}(\psi) - \sum_z t_z E_{k}(\varphi_z)\Bigr\} \geq 0.
\end{equation}
We extend this result considering a generic mixed state as target.
\begin{proposition} \label{prop:puretomixed}
	For $\psi \in \pure(AB)$ and $\sigma \in \md(AB)$, $\psi \xrightarrow\locc \sigma$ if and only if there exists a pure state decomposition $\Set{p_z, \chi_z}$ of $\sigma$ (i.e., $\sigma = \sum_z p_z \chi_z$) such that
	\begin{equation}
		\min_{k\in [|A|]} \Bigl\{E_{k}(\psi) - \sum_z p_z E_{k}(\chi_z)\Bigr\} \geq 0 .
	\end{equation}
\end{proposition}
That this condition is sufficient follows directly from Eq.~\eqref{eq:pure-to-prob-mixture}. That it is also necessary follows from the Lo-Popescu Theorem~\cite{LP01} which implies that if $\rho \xrightarrow{\locc} \sigma$, then there exists a pure state decomposition $\Set{p_z, \chi_z}$ of $\sigma$ that satisfies Eq.~\eqref{eq:pure-to-prob-mixture}.
The details of this proof can be found in Appendix~\ref{sec:pure-to-mix}.

For certain choices of $\rho\in \md(AB)$ and $\sigma\in\md(A'B')$, $\rho$ cannot be converted to $\sigma$ by $\locc$. It is then interesting to determine how well we can \textit{approximate} $\sigma$ with $\rho$ and $\locc$. To this end, one can associate with every distance $D$ defined on quantum states the conversion distance
\begin{equation}\label{eq:convdist}
	D(\rho \to \sigma) = \inf_{\mN\in\locc(AB \to A'B')}D(\mN(\rho), \sigma).
\end{equation}
This conversion distance determines how close to $\sigma$, with respect to the distance $D$, one can convert $\rho$ using only LOCC. A commonly used distance in quantum information is the trace distance $T(\sigma, \tau)= \frac{1}{2}\N{\sigma - \tau}_1$. The conversion distance associated with it is
\begin{equation}\label{eq:tr-convdist}
	T(\rho \to \sigma) = \inf_{\mN\in\locc(AB \to A'B')}\frac{1}{2}\N{\mN(\rho) - \sigma}_1.
\end{equation}
This conversion distance has an operational interpretation in terms of a result in state discrimination known as Holevo-Helstrom Theorem~\cite{Hel69, Hol73} (see Ref.~\cite[Theorem 3.4]{Wat18} for a review). 
Indeed, if a single copy of either $\sigma \in \md(AB)$ or $\mN(\rho)\in\md(AB)$, with $\mN\in\locc$, is given with equal probability, then the maximum probability $p_{\max}$ of correctly identifying the given state is bounded by
\begin{equation}
    p_{\max} \ge \frac{1}{2}(1 + T(\rho \to \sigma)).
\end{equation}
Moreover, for every fixed $\rho$ and $\sigma$, there always exists an $\mN\in\locc$ such that $p_{\max}$ is arbitrarily close to this lower bound. In this sense, $T(\rho \to \sigma)$ describes how well we can approximate $\sigma$ given access to $\rho$ and $\locc$.

Another distance used often in quantum information is the purified distance $P(\sigma, \tau) = \inf_{\psi_\sigma, \psi_\tau}T(\psi_\sigma, \psi_\tau)$~\cite{TCR10, GT20}, where the infimum runs over all purifications $\psi_\sigma$ and $\psi_\tau$ of $\sigma$ and $\tau$, respectively. In this case, the conversion distance is analogously defined as
\begin{equation}\label{eq:pur-convdist}
	P(\rho \to \sigma) = \inf_{\mN\in\locc(AB \to A'B')}P(\mN(\rho), \sigma).
\end{equation}

On pure states $\psi$, $\varphi \in \pure(AB)$, we define what we call star conversion distances via
\begin{equation} \label{eq:starconvdist}
	D_\star(\psi \to \varphi) = \min_{\vec{r}\succ \p}D(\vec{r}, \q),
\end{equation}
where $\p$, $\q \in \prob^\downarrow(|A|)$ are the Schmidt coefficients of $\psi$ and $\varphi$, respectively, and $D(\vec{r}, \q) = D(\diag(\vec{r}), \diag(\q))$. In general $D_\star(\psi \to \varphi) \neq D(\psi \to \varphi)$. However, we show in Appendix~\ref{sec:pur-dist} that in the case of the purified distance, the two conversion distances coincide on pure states.
\begin{theorem}
	Let $\psi$, $\varphi \in \pure(AB)$. Then $P(\psi \to \varphi) = P_\star(\psi \to \varphi)$.
\end{theorem}
This, in turn, is useful to show that $T(\psi \to \varphi)$ and $T_\star(\psi \to \varphi)$ are topologically equivalent on pure states, which means that if one can approximate $\varphi$ arbitrarily well with a sequence of states $\Set{\psi_n}$ in the sense that $\lim_{n \to \infty}T(\psi_n \to \varphi)=0$, then the same is true for $T_\star$, and vice-versa. This is guaranteed by the following result (Appendix~\ref{sec:star-conv-dist}).
\begin{lemma}\label{le:top-eq}
	Let $\psi$, $\varphi \in \pure(AB)$. Then
	\begin{equation}\label{equiv}
		\frac{1}{2}\left[T_\star (\psi \to \varphi)\right]^2 \leq T(\psi \to \varphi) \leq \sqrt{2T_\star(\psi \to \varphi)}.
	\end{equation}
\end{lemma}

We are particularly interested in $T_\star(\psi \to \varphi)$ because using tools from  approximate majorization~\cite{MOA11, NV01, Nie02, Tor70, Tor91, Ren16, HOS18}, one can derive a closed-form expression for it (Appendix~\ref{sec:star-conv-dist}).
 
\begin{theorem}\label{th:closeddist}
	Let $\psi$, $\varphi \in \pure(AB)$ and let $\p$, $\q \in \prob^\downarrow(|A|)$ be their corresponding Schmidt coefficients. Then,
	\begin{equation}
		T_\star(\psi \to \varphi) = \max_{k \in [\SR(\psi)]}\Set{\N{\p}_{(k)} - \N{\q}_{(k)}}.
	\end{equation}
	If $\xi\in \pure(AB)$ is separable and $\psi$, $\varphi$ are not,
	\begin{equation}
		T_{\star}(\psi \to \varphi)<T_\star(\xi \to \varphi).
	\end{equation}
\end{theorem}
This theorem provides, to the best of our knowledge, the first algorithm to compute a conversion distance in LOCC with a finite number of steps. Indeed, if $\psi$, $\varphi$, and their Schmidt coefficients are known, then one can compute $T_\star(\psi \to \varphi)$ using a finite memory, and a finite (perhaps very large) number of operations (additions or subtractions). This is not the case for the other conversion distances because they require a minimization over LOCC, which is in general unfeasible. The second part of the theorem shows that any entangled state is more useful in the approximation of all other entangled states than any separable state. 

Whilst Theorem~\ref{th:closeddist} is of independent interest, for example in the characterization of entanglement distillation and dilution~\cite{TFG23}, we discuss in the following how it yields new results concerning the embezzlement of entanglement. When referring to the conversion distance and star conversion distance, we will thus refer to the versions based on the trace distance from here on.

\section{Entanglement embezzlement}
As mentioned in the introduction, it is impossible to create additional entanglement with $\locc$ alone~\cite{VPRK97, PV07, CG19, CLP17}. If $\rho\in\md(A'B')$ is an entangled state, this implies that there cannot exist a $\psi\in\pure(AB)$ and a channel  $\mN\in\locc(AB\to A  BA'B')$ such that $\mN(\psi)=\psi\otimes\rho$, because this would increase the total amount of entanglement between  systems $AA'$ and $BB'$ with respect to any additive entanglement measure. It might however be possible to \textit{approximate} $\psi\otimes\rho$ in the sense that $T(\psi\to\psi\otimes\rho)\le\varepsilon$ for a small $\varepsilon$. In this case, it  is hard to distinguish $\psi\otimes\rho$ from the approximation. By keeping the systems $A'B'$, one would thus embezzle entanglement from the owner of $\psi$ - and if one would be able to make $\varepsilon$ arbitrarily small, it would be impossible to detect.

In fact, in Ref.~\cite{vDH03}, van Dam and Hayden showed that it is possible to  embezzle \textit{any} bipartite state $\sigma\in \md(A'B')$ arbitrarily well from a family of pure states $\Set{\chi_n^{AB}}_{n \in \mbN}$ in the sense that $\lim_{n\to\infty} T(\chi_n\to\chi_n\otimes\sigma)=0$. This implies that an arbitrarily good approximation of any $\sigma$ can be embezzled from $\chi_n$ whilst changing $\chi_n$ arbitrarily little, as long as $n$ is large enough. This motivates the following definition.
\begin{definition}\label{def:uef}
	A family of pure bipartite states $\Set{\chi_n}_{n\in \mbN}$ is called a \emph{universal embezzling family} if $\lim_{n \to \infty} T(\chi_n \to \chi_n \otimes \sigma)  = 0$ for every bipartite finite dimensional state $\sigma$.
\end{definition}
This definition is very similar to the one provided in Refs.~\cite{vDH03, LW14}. The main difference is that these works only consider protocols using local operations (LO), while in this work, we allow for classical communication too. The set of operations that we consider is therefore larger, and as a result, if a family of states is not an embezzling family according to our definition, then it is not an embezzling family in the sense of Refs.~\cite{vDH03, LW14}. Surprisingly, the original embezzling family proposed in Ref.~\cite{vDH03} is rather unique in a sense that we will specify later, even when we allow for classical communication. Another difference with Refs.~\cite{vDH03, LW14} is that in those works the fidelity was used to quantify the conversion error. Since the fidelity and trace distance are topologically equivalent, this is irrelevant.

The term `universal' in Definition~\ref{def:uef} underlines the property that any bipartite state can be embezzled. Since for every bipartite state $\sigma \in \md(AB)$, where $A = B = m$, it is possible to LOOC-convert the maximally entangled state $\Phi_m$ into $\sigma$ (see, e.g., Proposition~\ref{prop:puretomixed}), it is enough to check whether $\lim_{n \to \infty} T(\chi_n \to \chi_n \otimes \Phi_m)  = 0$ for all $m \in \mbN$ to determine if $\Set{\chi_n}_{n \in \mbN}$ is a universal embezzling family: This follows for example by combining Lemma~\ref{le:top-eq} and the triangular inequality for the star conversion distance proven in Appendix~\ref{sec:star-conv-dist}. Even simpler, it is in fact equivalent to only require $\lim_{n \to \infty} T(\chi_n \to \chi_n \otimes \Phi_2)  = 0$~\cite[Lemma~2 for the case of embezzlement with LO]{LW14}, because if one can embezzle enough copies of $\Phi_2$, then one can convert them into $\Phi_m$ with LOCC (see Appendix~\ref{sec:emb-fam} for more details). It is also important to note that technically, one could have required that $\liminf_{n \to \infty} T(\chi_n \to \chi_n \otimes \sigma)  = 0$, since this would also allow to embezzle any state arbitrarily well. However, since one can always choose a subfamily, we decided to keep the definition in line with Ref.~\cite{LW14}. Lastly, due to Eq.~\eqref{equiv}, we can replace $T(\chi_n \to \chi_n \otimes \sigma)$ in Definition~\ref{def:uef} with $T_\star(\chi_n \to \chi_n \otimes \Phi_2)$. Thanks to the closed formula in Theorem~\ref{th:closeddist}, we obtain the following complete characterization of universal embezzling families (Appendix~\ref{sec:emb-fam}).
\begin{theorem}\label{thm:charEmbFam}
	A family of pure bipartite states $\Set{\chi_n}_{n\in \mbN}$ with corresponding Schmidt coefficients $\Set{\p^{(n)}}_{n \in\mbN}$ is a universal embezzling family if and only if
	\begin{equation}\label{eq:charEmbFam}
		\lim_{n \to \infty} \max_{l \in [A_n]}\Set{\N{\p^{(n)}}_{{(2l -1)}} - \N{\p^{(n)}}_{{(l-1)}} } = 0,
	\end{equation}
	where $A_n = \ceil{\SR(\chi_n)/2}$ is the ceiling of half the Schmidt rank of $\chi_n$.
\end{theorem}
The problem of determining if a family of states is a universal embezzling family has therefore been restated as a rather simple optimization problem, which in many cases can be solved numerically or even analytically. By choosing $l=1$ in Eq.~\eqref{eq:charEmbFam} one obtains the following necessary condition for a universal embezzling family.
\begin{corollary}[{cf. Ref.~\cite[Lemma~3 (LO)]{LW14}}]
	If a family of pure bipartite states $\Set{\chi_n}_{n \in \mbN}$ is a universal embezzling family, then $\lim_{n \to \infty} p_1^{(n)} =0$, where $\p^{(n)} \in \prob^\downarrow (d_n)$ are the Schmidt coefficients of $\chi_n$.
\end{corollary}
This implies that if a family of pure bipartite states $\Set{\chi_n}_{n \in \mbN}$ is a universal embezzling family, then $\lim_{n \to \infty} \SR(\chi_n)=+\infty$, where $\SR(\chi_n)$ is the Schmidt rank of $\chi_n$. Indeed, considering that the entries of any $\p^{(n)}$ sum to one, if the largest entry converges to zero, then the number of non-zero entries must diverge.

When looking for candidates for universally embezzling families, it is common~\cite{vDH03,LW14} to consider families of bipartite states
\begin{equation}\label{eq:famstates}
	\Ket{\chi_n} = \frac{1}{\sqrt{F_n}}\sum_{x=1}^n \sqrt{f(x)} \Ket{xx}_{AB},
\end{equation}
defined by a function $f\colon \mbN \to \mbR^+$, where $F_n = \sum_{x = 1}^n f(x)$ is a normalization constant. With this choice, one ensures that the families of states have a common structure, i.e., $\chi_{m>n}$ is obtained from $\chi_n$ by appending additional coefficients and renormalizing. By choosing $f(x)=x^{-1}$, one recovers the universal embezzling family introduced by van Dam and Hayden~\cite{vDH03}. This family is rather unique: If we assume that $f$  has a reasonable asymptotic behavior, specifically that it is asymptotically non-increasing and  $f(x)/ x^\alpha$ is asymptotically monotonic for all $\alpha\in\mbR$, then the family of states $\Set{\chi_n}$ is an embezzling family if and only if $f$ is asymptotically close to $x^{-1}$ in the sense that for every $\varepsilon>0$,
\begin{equation}\label{eq:subpoly}
	\lim_{x\to \infty}\frac{f(x)}{x^{-1-\varepsilon}} =  \infty \, , \quad \lim_{x\to \infty}\frac{f(x)}{x^{-1+\varepsilon}} = 0.
\end{equation}
As we will discuss now, the assumptions on the asymptotic behavior of $f$, which we require for technical reasons,  are not particularly restrictive. In entanglement theory, it is possible to consider, w.l.o.g., only states with non-increasing Schmidt coefficients. With the first assumption, we require that the family of states, at least asymptotically, has this behavior. The second assumption rules out functions that are asymptotically monotonic but have oscillating components, for an example of such a function see Appendix~\ref{sec:uniq-vdh}, Eq.~\eqref{eq:ex-non-incr-f}.

In Ref.~\cite{LW14}, it was shown that certain functions that differ from $x^{-1}$ by logarithmic factors lead to universal embezzling families too. Since these functions satisfy the constraints in Eq.~\eqref{eq:subpoly}, this is in accordance with our findings. In addition, in Appendix~\ref{sec:uniq-vdh}, we also show that many non-decreasing functions do not lead to universal embezzling families (see Propositions~\ref{prop:f-incr-vdh-alpha} and~\ref{prop:f-incr-vdh-k} for details).

A special case of the functions discussed so far are the functions $f_\alpha(x) = x^\alpha$ with $\alpha \in \mbR$. For the families of states $\Set{\chi_n^\alpha}$ generated by such functions, we analytically compute the exact value of $\lim_{n \to \infty}T_\star(\chi_n^\alpha \to \chi_n^\alpha \otimes \Phi_m)$ for $\alpha \geq -1$ and lower and upper bound it for $\alpha < -1$. The details of the computations can be found in Appendix~\ref{sec:gen-van-dam} and the results for $m =2$ are shown in Figure~\ref{fig:falpha}. Clearly, the limit of the star conversion distance is zero only for $\alpha = -1$, showing again the uniqueness of the choice made by van Dam and Hayden amongst the functions $f(x)=x^\alpha$.
\begin{figure}[htp]\centering
	\includegraphics[width=0.5\columnwidth]{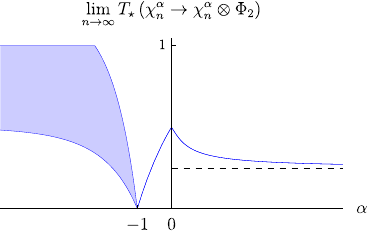}
	\caption{\textit{Uniqueness of the van Dam and Hayden embezzling family ---}
		The family of states $\Set{\chi^\alpha_n}_{n \in \mbN}$ introduced in the main text is a universal embezzling family if and only if $\alpha=-1$. This can be seen in the above plot showing the analytically derived exact value of $\lim_{n\to \infty}T_\star(\chi_n^\alpha \to \chi_n^\alpha \otimes \Phi_2)$ for $\alpha \ge -1$ and lower and upper bounds for $\alpha < -1$.}
	\label{fig:falpha} 
\end{figure}

\section{Conclusions}
In Proposition~\ref{prop:puretomixed} we provided necessary and sufficient conditions for a deterministic $\locc$-conversion from a pure bipartite state to a mixed bipartite state. This extends the results already known for pure to pure state $\locc$-conversions, whether deterministic or probabilistic~\cite{Nie99,JP99A}. We then exploited this result to prove the topological equivalence of the newly defined star conversion distance between pure states and the trace conversion distance commonly used in literature. The star conversion distance exhibits a closed formula (Theorem~\ref{th:closeddist}). This is remarkable, since the mixed state $\locc$-conversion problem is NP-hard~\cite{Gur03}.

The closed formula in turn allowed us to completely characterize universal embezzling families (see Definition~\ref{def:uef} and Refs.~\cite{vDH03,LW14} for an analogous definition for LO) in terms of a simple optimization problem stated in  Theorem~\ref{thm:charEmbFam}.
With this characterization at hand, we discussed the uniqueness of the van Dam and Hayden family~\cite{vDH03}. For specific families of states generalizing the van Dam and Hayden family, we explicitly evaluated their star conversion distance to maximally entangled states and showed that they are only universally embezzling if they are exactly the van Dam and Hayden family, see Figure~\ref{fig:falpha}. Therefore, the van Dam and Hayden embezzling family shows unique properties even for protocols that involve classical communication, as already noticed in Ref.~\cite{vDH03}. This suggests a direction for future work, namely, to determine whether LOCC embezzlement implies LO embezzlement (the other direction is trivial), and therefore to investigate if classical communication is relevant in entanglement embezzlement or not. It is worth noting that so far, the research on entanglement embezzlement focuses solely on families composed of pure states. A more comprehensive theory of embezzlement that includes families of mixed states is yet to be developed.

Originally introduced in the resource theory of entanglement as a generalization of catalysis~\cite{vDH03}, embezzlement has recently also been studied in other resource theories including non-uniformity~\cite{NMCJW15, GMNSYH15}, coherence~\cite{BCP14, BDG15, CZZZ19, ST22}, and athermality~\cite{BHORS13, HJ13, BHNOW15, NMCJW15, LBS21}. Moreover, fundamental limits for embezzlement have been proved in Ref.~\cite{RT22} and applied to the resource theories mentioned above. The aforementioned resource theories are related to the resource theory of entanglement via majorization, which is the tool that we used to derive the closed formula for the star conversion distance. A natural next step is to investigate whether it is possible to derive a similar formula in these other majorization-based resource theories. However, this task is not trivial: In the resource theory of non-uniformity, the majorization relation is inverted, and the free state is the maximally mixed state, which is fundamentally different from the free states in the resource theory of entanglement; in the resource theory of athermality, state conversion is described by relative majorization, which is a generalization of majorization. It is worth mentioning that in Ref.~\cite{BHNOW15} it has been shown that embezzlement of athermality allows to violate the second law of thermodynamics. However, by using the work distance~\cite{BHNOW15} as conversion distance, or by imposing physical constraints on the catalyst~\cite{NMCJW15}, e.g., finite dimension or finite energy expectation value, athermality embezzlement is no longer possible, thus restoring the validity of the second law of thermodynamics.

As described in the introduction, universal embezzling families have been useful in many applications~\cite{BDHSW14, BCR11, RV14, Gro53, Pis12,LTW13, DSV15, RV15, JLV20}. Our results provide an easy way to check whether a family is universally embezzling. Moreover, we ruled out large classes of potential candidates that are fundamentally different from the van Dam and Hayden family. The complete characterization of universally embezzling families therefore contributes to a more efficient usage of entanglement in technology.

\begin{acknowledgments}
	The authors thank an anonymous reviewer for suggesting a simpler proof of Theorem~\ref{th:closeddist}, which we included in this version of the manuscript. The authors acknowledge support from the Natural Sciences and Engineering Research Council of Canada (NSERC). E.~Z.\ acknowledges support from the Alberta Graduate Excellence Scholarship (AGES) and from the Eyes High International Doctoral Scholarship. T. T. acknowledges support from the Pacific Institute for the Mathematical Sciences (PIMS). The research and findings may not reflect those of the Institute.
\end{acknowledgments}

\bibliographystyle{quantum}

\newpage
\appendix
\counterwithin*{equation}{section}
\renewcommand{\theequation} {\thesection\arabic{equation}}

%% Theorems App
\newcounter{thm}[section]
\renewcommand{\thethm}{\thesection.\arabic{thm}}
\newtheorem{lemma-app}[thm]{Lemma}
\newtheorem{theorem-app}[thm]{Theorem}
\newtheorem{proposition-app}[thm]{Proposition}
\newtheorem{corollary-app}[thm]{Corollary}
\newtheorem{definition-app}[thm]{Definition}

\setcounter{theorems}{0}

\section{Pure to Mixed State Conversions with LOCC}\label{sec:pure-to-mix}
{In this section, we prove Proposition}~\ref{prop:puretomixed} of the main text, which we restate for readability. 
\begin{proposition}
	For $\psi \in \pure(AB)$ and $\sigma \in \md(AB)$, $\psi \xrightarrow\locc \sigma$ if and only if there exists a pure state decomposition $\Set{p_z, \chi_z}$ of $\sigma$ (i.e., $\sigma = \sum_z p_z \chi_z$) such that
	\begin{equation}\label{eq:pure-to-mix-cond}
		\min_{k\in [|A|]} \Bigl\{E_{k}(\psi) - \sum_z p_z E_{k}(\chi_z)\Bigr\} \geq 0 .
	\end{equation}
\end{proposition}

\begin{proof}
	We first assume that $\psi$ can be converted into $\sigma$ with $\locc$ operations. Then, according to Refs.~\cite{LP01, CS15}, there exists a protocol in which Alice performs a generalized measurement $\Set{M_z}$ and Bob performs a unitary transformation $U_z$ conditioned on the measurement's outcome such that $\sigma =\sum_{z} (M_z \otimes U_z)\psi (M_z \otimes U_z)^\dagger$. If Alice and Bob record the outcome of the measurement in a classical system $X$, the output of the protocol is $\sum_{z} p_z {\KBra{z}{z}}^X \otimes \chi_z^{AB} \coloneqq \Set{p_z, \chi_z}$, where
	\begin{equation}
		\Ket{\chi_z} = \frac{(M_z \otimes U_z)\Ket{\psi}}{\N{(M_z \otimes U_z)\Ket{\psi}}} \, , \quad p_z = \N{(M_z \otimes U_z)\Ket{\psi}}^2.
	\end{equation}
	This shows that Alice and Bob can convert $\psi$ into the ensemble $\Set{p_z, \chi_z}$ with $\locc$ operations, which is equivalent to the condition~\cite{Vid99, JP99A}
	\begin{equation}\label{eq:ext-nie}
		\min_{k \in [d]} \Bigl(E_{k}(\psi) - \sum_{z}p_z E_{k}(\chi_z)\Bigr)\geq 0,
	\end{equation}
	where $E_{k}$ was introduced in the main text and $d=|A|=|B|$ as per our convention. This proves the necessary condition.
	
	For the reverse, we assume that $\Set{p_z, \chi_z}$ is a pure state decomposition of $\sigma$ that satisfies
	\begin{equation}
		\min_{k \in [d]} \Bigl(E_{k}(\psi) - \sum_{z} p_z E_{k}(\chi_z)\Bigr)\geq 0.
	\end{equation}
	This is equivalent to $\psi \xrightarrow{\locc} \Set{p_z, \chi_z} = \sum_{z} p_z {\KBra{z}{z}}^X \otimes \chi_z^{AB}$~\cite{Vid99, JP99A}. Alice and Bob can trace out the classical system and they obtain $\sum_{z} p_z \chi_z = \sigma$, thus $\psi$ can be converted into $\sigma$ with $\locc$ operations.
\end{proof}

\section{Purified Conversion Distance}\label{sec:pur-dist}
In this section, we present some results concerning the purified conversion distance~\cite{TCR10,GT20}, which are useful to prove the theorems about the star conversion distance presented in the main text. The purified distance between two states $\rho$, $\sigma \in \md\left(AB\right)$ is defined as
\begin{equation}
	P\left(\rho, \sigma\right) = \sqrt{1 - F^2\left(\rho, \sigma\right)},
\end{equation}
where $F\left(\rho, \sigma\right) = \N{\sqrt{\rho}\sqrt{\sigma}}_1$ is the fidelity.
The purified distance is a metric and according to Uhlmann’s Theorem~\cite{Uhl76},
\begin{equation}\label{eq:pur-dist-min-ext}
	P\left(\rho, \sigma\right) = \min_{\psi, \varphi}T\left(\psi, \varphi\right),
\end{equation}
where $T\left(\rho, \sigma\right)=\frac{1}{2}\N{\rho-\sigma}_1$ is the trace distance and the minimization runs over all purifications $\psi$ and $\varphi$ of $\rho$ and $\sigma$, respectively. Importantly, the purified distance is topologically equivalent to the trace distance~\cite{TCR10},
\begin{equation}\label{eq:top-equiv}
	T\left(\rho, \sigma\right) \leq P\left(\rho, \sigma\right) \leq \sqrt{2 T\left(\rho, \sigma\right)}.
\end{equation}
We now recall the definition of the purified conversion distance and purified star conversion distance introduced in the main text,
\begin{equation}
	P\left(\rho \to \sigma\right) = \inf_{\mN\in \locc} P\left(\mN\left(\rho\right), \sigma\right)\, , \quad P_\star\left(\psi \to \varphi\right) = \min_{\vec{r} \succ \p} P\left(\vec{r}, \q\right),
\end{equation}
where $\p$, $\q \in \prob^\downarrow(d)$ are the Schmidt coefficients of $\psi$ and $\varphi$, respectively, and $P\left(\vec{r}, \q\right) = \sqrt{1 - F^2(\vec{r}, \q)}$ is the classical version of the purified distance with 
\begin{equation}
	F(\vec{r}, \q)\coloneqq F(\diag(\vec{r}),\diag(\q))=\sum_x \sqrt{r_x q_x}.
\end{equation}

We next prove that the two purified conversion distances recalled above are equal on pure states. To this end, we need the following Lemma.
\begin{lemma-app}\label{le:convex}
	Let $\q \in \prob(d)$. The function $f\colon \prob(d) \to [0,1]$ defined as $f(\vec{v}) = \left(\sum_{x=1}^d\sqrt{q_x v_x}\right)^2$ is concave.
\end{lemma-app}
\begin{proof}
	For details about concave functions, see~Ref.~\cite{RV73}. Here, we are going to use that a twice differentiable function is concave if and only if its Hessian matrix is negative semi-definite. The  functions $f_{x,y}(v_x,v_y):=\sqrt{q_xv_x}\sqrt{q_yv_y}$ are twice differentiable in $v_x,v_y$ for $v_x,v_y\in(0,1]$ and their Hessian matrices are given by
	\begin{equation}
		H\left(v_x, v_y\right) = \frac{1}{4}\begin{bmatrix}
			-\sqrt{\frac{q_xq_yv_y}{v_x^3}} & \sqrt{\frac{q_xq_y}{v_xv_y}}\\
			\sqrt{\frac{q_xq_y}{v_xv_y}} & -\sqrt{\frac{q_xq_yv_x}{v_y^3}}
		\end{bmatrix}.
	\end{equation}
	Since $\det H(v_x,v_y) =0$, one of the eigenvalues of $H(v_x, v_y)$ is $0$ and the other is equal to the trace of $H(v_x, v_y)$, which is smaller than or equal to zero. This implies that the functions $f_{x,y}(v_x,v_y)$ are concave for $v_x,v_y\in(0,1]$. Moreover, since $0=f_{x,y}(0,v_y)=f_{x,y}(v_x,0)=f_{x,y}(0,0)$, it is easy to see that $f_{x,y}(v_x,v_y)$ is in fact concave for $v_x,v_y\in[0,1]$. From this follows that $f$ is concave, since 
	\begin{equation}
		f\left(\vec{v}\right) = \sum_x v_x q_x + \sum_y \sum_{x\neq y}\sqrt{q_xv_x}\sqrt{q_yv_y}
	\end{equation}
	is the sum of concave functions.
\end{proof}
We are now ready to prove the promised theorem. 
\begin{theorem}\label{th:pur-dist}
	Let $\psi, \varphi \in \pure(AB)$, then $P\left(\psi \to \varphi\right) = P_\star\left(\psi \to \varphi\right)$.
\end{theorem}
\begin{proof}
	We assume w.l.o.g. that all pure states are in standard form, that is, $\Ket{\psi} = \sum_x \sqrt{p_x} \Ket{xx}_{AB}$, where $\p \in \prob^\downarrow(d)$ and $\Set{\Ket{x}_A}$ and $\Set{\Ket{x}_B}$ are fixed bases for $A$ and $B$, respectively (see main text for more details). Let $\p,\q\in \prob^\downarrow(d)$ be the Schmidt coefficients or $\psi$ and $\varphi$, respectively. With $\vec{r} \in \prob^\downarrow(d)$, as a consequence of Nielsen's Theorem~\cite{Nie99}, $\vec{r} \succ \p$ if and only if there exists an $\mN \in \locc$ such that $\mN(\psi)\in \pure(AB)$ has Schmidt coefficients $\vec{r}$. We notice that
	\begin{equation} \label{eq:purdistclass}
		\begin{aligned}
			P(\mN(\psi), \varphi) &= \sqrt{1 - F^2(\mN(\psi), \varphi)} = \sqrt{1 - \A{\Braket{\mN(\psi)|\varphi}}^2} \\
			&= \sqrt{1 - \left(\sum_x \sqrt{r_x q_x}\right)^2} = \sqrt{1 - F^2(\vec{r}, \q)} \\
			&= P(\vec{r}, \q)  .
		\end{aligned}
	\end{equation}
	This implies that
	\begin{equation}\label{eq:first-dist-ineq}
		\begin{aligned}
			P_\star\left(\psi \to \varphi\right) &= \min_{\vec{r} \succ \p} P\left(\vec{r}, \q\right) \\
			&= \min_{\substack{\mN\in \locc \\ \mN(\psi) \in \pure(AB)}} P\left(\mN\left(\psi\right), \varphi\right) \\
			&\geq \inf_{\mN\in \locc} P\left(\mN\left(\psi\right), \varphi\right) \\
			&= P\left(\psi \to \varphi\right).
		\end{aligned}
	\end{equation}
	
	This inequality follows from the definition of the star purified conversion distance because the minimization in the star conversion distance is done over a smaller set. The non-trivial part is to show that the opposite inequality holds as well. To this end, we want to show that for every mixed state $\sigma$ such that $\psi \xrightarrow{\locc} \sigma$, there exists a pure state $\chi$ such that $\psi \xrightarrow\locc \chi$ and $P(\sigma, \varphi) \geq P(\chi, \varphi)$. It is then sufficient to consider only pure output states for the computation of the purified conversion distance, which implies the reverse inequality.
	
	Let $\sigma = \mM(\psi) \in \md(AB)$, where $\mM \in \locc$. Also, let $\Set{t_z, \chi_z}$ be a pure state decomposition of $\sigma$ that satisfies Eq.~\eqref{eq:pure-to-mix-cond} (where the $\chi_z$ are not necessarily in standard form), and $\s^{(z)}$ be the Schmidt coefficient of $\chi_z$ for every $z$.
	Furthermore, for every $\chi_z$, let $\tilde{\chi}_z$ be the pure state in standard form that is equal to $\chi_z$ up to local unitaries, let $\tilde{\sigma}=\sum_{z}t_z\tilde{\chi}_z$,
	and define $\chi \in \pure(AB)$ as the pure bipartite state (in standard form) with Schmidt coefficients $\s = \sum_z t_z \s^{(z)}$. We notice that for all $k \in [d]$
	\begin{equation} \label{eq:chi-psi-monotone}
		\begin{aligned}
			E_{k}(\chi) &= 1-\sum_{x=1}^k s_x = 1-\sum_{x =1}^k \sum_{z} t_z s_x^{(z)} \\
			&= 1-\sum_z t_z(1- E_{k}(\chi_z)) = \sum_z t_zE_{k}(\chi_z)\\
			&\leq E_{k}(\psi) ,
		\end{aligned}		
	\end{equation}
	where the last inequality follows from the fact that $\Set{t_z, \chi_z}$ satisfies Eq.~\eqref{eq:pure-to-mix-cond}. This implies that $\s \succ \p$.
	
	The next step is to show that $P(\sigma, \varphi) \geq P(\s, \q)$. First, we note that due to the von Neumann trace inequality~\cite{Neu37, Mir75}, 
	\begin{equation}\label{eq:fid1}
		F^2\left(\tilde{\sigma}, \varphi\right) = \sum_{z}t_z \left(\sum_x \sqrt{q_x s_x^{(z)}}\right)^2\ge F^2\left(\sigma, \varphi\right) 
	\end{equation}
	and 
	\begin{equation}\label{eq:fid2}
		F^2(\chi, \varphi) = F^2\left(\s, \q\right) = \left(\sum_x \sqrt{q_x\sum_zt_z s_x^{(z)}}\right)^2 .
	\end{equation}
	Second, we introduce the concave function $f\left(\vec{v}\right) = \left(\sum_x \sqrt{q_x v_x}\right)^2$, for $\vec{v} \in \prob\left(d\right)$ (see Lemma~\ref{le:convex}), and rewrite Eq.~\eqref{eq:fid1} and Eq.~\eqref{eq:fid2} as
	\begin{equation}
		F^2\left(\sigma, \varphi\right) \le F^2\left(\tilde{\sigma}, \varphi\right) = \sum_{z}t_z f\left(\s^{(z)}\right) \, , \quad F^2\left(\s, \q\right) = f\left(\sum_z t_z \s^{(z)}\right) .
	\end{equation}
	Finally, the concavity of $f$ and the multidimensional Jensen's inequality~\cite{Fer67} imply
	\begin{equation}
		F^2\left(\sigma, \varphi\right) \le F^2\left(\tilde{\sigma}, \varphi\right) = \sum_{z}t_z f\left(\s^{(z)}\right) \leq f\left(\sum_z t_z \s^{(z)}\right) = F^2\left(\s, \q\right) ,
	\end{equation}
	which is equivalent to $P(\sigma,\varphi)\geq P(\s, \q)$. So far, we have shown that for every $\sigma=\mM(\psi)$, with $\mM\in \locc$, there exists an $\s^{(\mM(\psi))} \in \prob^\downarrow(d)$ such that $\s^{(\mM(\psi))} \succ \p$ and $P(\sigma,\varphi)\geq P(\s^{(\mM(\psi))}, \q)$, thus
	\begin{equation}\label{eq:second-dist-ineq}
		\begin{aligned}
			P(\psi \to \varphi) &= \inf_{\mM\in \locc} P\left(\mM\left(\psi\right), \varphi\right)\\
			&\geq \inf_{\mM\in \locc} P\left(\s^{(\mM(\psi))}, \q\right) \\
			&\geq \min_{\vec{r} \succ \p} P\left(\vec{r}, \q\right) \\
			&= P_\star\left(\psi \to \varphi\right) .
		\end{aligned}
	\end{equation}
	Eq.~\eqref{eq:first-dist-ineq} and Eq.~\eqref{eq:second-dist-ineq} imply that $P(\psi \to \varphi) =P_\star\left(\psi \to \varphi\right)$, and this concludes the proof.
\end{proof}

\section{The Star Conversion Distance}\label{sec:star-conv-dist}
In this section, we discuss properties of the star conversion distance based on the trace distance and provide proofs omitted in the main text. We begin by proving that the trace star conversion distance is topologically equivalent to the standard conversion distance defined via the trace distance. 
\begin{lemma}
	Let $\psi$, $\varphi \in \pure(AB)$. Then
	\begin{equation}
		\frac{1}{2}\left[T_\star (\psi \to \varphi)\right]^2 \leq T(\psi \to \varphi) \leq \sqrt{2T_\star(\psi \to \varphi)}.
	\end{equation}
\end{lemma}
\begin{proof}
Using Eq.~\eqref{eq:top-equiv} and Theorem~\ref{th:pur-dist}, we obtain\begin{equation}
	\begin{aligned}
		T(\psi \to \varphi) &= \inf_{\mN \in \locc} T(\mN(\psi), \varphi) \\
		&\leq \inf_{\mN \in \locc} P(\mN(\psi), \varphi) = P(\psi \to \varphi) = P_\star(\psi \to \varphi)= \min_{\vec{r} \succ \p} P(\vec{r},\q) \\
		&\leq \min_{\vec{r} \succ \p}\sqrt{2T(\vec{r}, \q)} = \sqrt{2T_\star(\psi \to \varphi)},
	\end{aligned}
\end{equation}
where $T(\p, \q) \coloneqq T(\diag(\p), \diag(\q))$. Analogously,
\begin{equation}
	\begin{aligned}
		T_\star(\psi \to \varphi) &=\min_{\vec{r} \succ \p}T(\vec{r}, \q) \\
		&\leq \min_{\vec{r} \succ \p} P(\vec{r},\q) = P_\star(\psi \to \varphi) = P(\psi \to \varphi) = \inf_{\mN \in \locc} P(\mN(\psi), \varphi)   \\
		&\leq  \inf_{\mN \in \locc} \sqrt{2T(\mN(\psi), \varphi)} = \sqrt{2T(\psi \to \varphi)}.
	\end{aligned}
\end{equation}
\end{proof}
This proves that $T_\star\left(\psi \to \varphi\right)$ is topologically equivalent to $T\left(\psi \to \varphi\right)$.

In the following, we provide the proof of Theorem~\ref{th:closeddist} of the main text, which we restate below for readability. In other words, we derive a closed formula for $T_\star\left(\psi \to \varphi\right)$ based on the Schmidt coefficients $\p,\q \in \prob^\downarrow(d)$ of $\psi$ and $\varphi$, respectively. 

\begin{theorem}\label{th:closeddist-app}
	Let $\psi$, $\varphi \in \pure(AB)$ and let $\p$, $\q \in \prob^\downarrow(|A|)$ be their corresponding Schmidt coefficients. Then,
	\begin{equation}
		T_\star(\psi \to \varphi) = \max_{k \in [\SR(\psi)]}\Set{\N{\p}_{(k)} - \N{\q}_{(k)}}.
	\end{equation}
	If $\xi\in \pure(AB)$ is separable and $\psi$, $\varphi$ are not, then 
	\begin{equation}
		T_\star(\psi \to \varphi)<T_\star(\xi \to \varphi).
	\end{equation}
\end{theorem}

\begin{proof}
	Let $\varepsilon \in [0,1]$ and 
	$\mb^\varepsilon_\q = \Set{\q'   :   \frac{1}{2}\N{\q - \q'}_1 \le \varepsilon}$. As shown in Ref.~\cite{HOS18}, there exist probability vectors   
	$\bar\q^{(\varepsilon)} \in \mb^\varepsilon_\q$ called steepest $\varepsilon$-approximations of $\q$ such that $\bar\q^{(\varepsilon)} \succ \q'$ for all $\q' \in \mb^\varepsilon_\q$. Moreover, these steepest $\varepsilon$-approximations can be constructed explicitly: If $ \frac{1}{2}\N{\q - \e_1}_1 \le \varepsilon$, then $\bar\q^{(\varepsilon)} = \e_1$, otherwise let $k_\varepsilon \in [d]$ be the index satisfying
	\begin{equation}\label{eq:k}
		\sum_{x = 1}^{k_\varepsilon} q_x \leq 1 - \varepsilon \quad \text{and} \quad \sum_{x = 1}^{k_\varepsilon +1} q_x > 1- \varepsilon   .
	\end{equation}
	The components of $\bar\q^{(\varepsilon)}$ are then given by
	\begin{equation} \label{eq:steepestEpsilonApprox}
		\bar q^{(\varepsilon)}_x =
		\begin{cases}
			q_1 + \varepsilon & \text{if }   x = 1,\\
			q_x & \text{if }   x \in \Set{2, \dots, k_\varepsilon},\\
			1 - \varepsilon - \sum_{x = 1}^{k_\varepsilon} q_x & \text{if }   x = k_\varepsilon +1,\\
			0 & \text{otherwise}.
		\end{cases}  
	\end{equation}
	We now show that 
	\begin{equation}
		\min_{\vec{r} \succ \p}\Set{\frac{1}{2}\N{\q - \vec{r}}_1} = \min\Set{\varepsilon \in [0,1]  :   \bar\q^{(\varepsilon)} \succ \p}.
	\end{equation}
	First, we notice that if $\vec{r}_\star$ is an optimizer of $\min_{\vec{r} \succ \p}\Set{\frac{1}{2}\N{\q - \vec{r}}_1} $ and $\tilde\varepsilon = \frac{1}{2}\N{\q- \vec{r}_\star}_1$, then $\vec{r}_\star \in \mb^{\tilde\varepsilon}_\q$ and therefore $\bar\q^{(\tilde\varepsilon)} \succ \vec{r}_\star$. By transitivity, we also have $\bar\q^{(\tilde\varepsilon)} \succ \p$, which implies that
	\begin{equation}
		\min\Set{\varepsilon \in [0,1]  :   \bar\q^{(\varepsilon)} \succ \p} \leq \tilde\varepsilon = \min_{\vec{r} \succ \p}\Set{\frac{1}{2}\N{\q - \vec{r}}_1}.
	\end{equation}
	For the reverse inequality, let $\varepsilon_\star$ be the optimizer of $\min\Set{\varepsilon \in [0,1]  :   \bar\q^{(\varepsilon)} \succ \p} $. By definition, $\bar\q^{(\varepsilon_\star)}\succ\p$ and $\frac{1}{2}\N{\q -\bar\q^{(\varepsilon_\star)} }_1 \le \varepsilon_\star$, thus 
	\begin{equation}
		\min_{\vec{r} \succ \p}\Set{\frac{1}{2}\N{\q - \vec{r}}_1} \le \varepsilon_\star = \min\Set{\varepsilon \in [0,1]  :   \bar\q^{(\varepsilon)} \succ \p}.
	\end{equation}
	This shows that
	\begin{equation}\label{eq:d-star-min}
		T_\star(\psi \to \varphi) = \min_{\vec{r} \in \prob^\downarrow(d)}\Set{\frac{1}{2}\N{\q - \vec{r}}_1   :   \vec{r} \succ \p} = \min\Set{\varepsilon \in [0,1]  :   \bar\q^{(\varepsilon)} \succ \p}.
	\end{equation}
	
	We observe that $\N{\bar\q^{(\varepsilon)}}_{(k)} = \min\Set{1, \N{\q}_{(k)} + \varepsilon}$, thus the condition $\bar\q^{\varepsilon}\succ \p$ is equivalent to
	\begin{equation}
		\N{\p}_{(k)} \leq \min\Set{1, \N{\q}_{(k)} + \varepsilon}\,, \quad \forall k \in [d].
	\end{equation}
	This expression is further simplified by noticing that $\p$ is a probability vector, and therefore $\N{\p}_{(k)} \leq 1$ for all $k \in [d]$. Consequently,
	\begin{equation}
		\bar\q^{\varepsilon} \succ \p \quad \Leftrightarrow\quad \N{\p}_{(k)} \leq \N{\q}_{(k)} + \varepsilon,\, \forall k \in [d] \quad \Leftrightarrow\quad \max_{k \in [d]} \Set{\N{\p}_{(k)} - \N{\q}_{(k)}} \leq \varepsilon.
	\end{equation}
	In combination with the minimization in Eq.~\eqref{eq:d-star-min} follows that
	\begin{equation}\label{eq:stardist}
		T_\star(\psi \to \varphi) = \max_{k \in [d]} \Set{\N{\p}_{(k)} - \N{\q}_{(k)}}.
	\end{equation}
	To conclude the proof of the first part of the theorem, we observe that for $k \geq \SR(\psi^{AB})$, $\N{\p}_{(k)} = 1$ and $\N{\q}_{(k)}$ is non-decreasing with $k$, thus we can restrict the maximization to $k \leq \SR(\psi^{AB})$. 
	
	For the second part, we notice that according to Eq.~\eqref{eq:d-star-min},
	\begin{equation}
		T_\star(\xi \to \varphi) = \min_{\vec{r} \in \prob^\downarrow(d)}\Set{\frac{1}{2}\N{\q - \vec{r}}_1   :   \vec{r} \succ \e_1} =\frac{1}{2}\N{\q - \e_1}_1 =1 - q_1.
	\end{equation}
	Furthermore, from $\e_1 \succ \p$ and the transitivity of the majorization-relation, it follows that
	\begin{equation}\label{eq:conv-dist-ineq}
		\begin{aligned}
			T_\star(\psi \to \varphi)&=\min_{\vec{r} \in \prob^\downarrow(d)}\Set{\frac{1}{2}\N{\q - \vec{r}}_1   :   \vec{r} \succ \p} \\
			&\leq \min_{\vec{r} \in \prob^\downarrow(d)}\Set{\frac{1}{2}\N{\q - \vec{r}}_1   :   \vec{r} \succ \e_1}\\
			&= T_\star(\xi \to \varphi).
		\end{aligned}
	\end{equation}
	
	To rule out equality, suppose that $T_\star(\psi \to \varphi) = \max_{k \in [d]} \Set{\N{\p}_{(k)} - \N{\q}_{(k)}} =1-q_1$, and denote with $k_\star$ an index that achieves this maximum. Then
	\begin{equation}
		\N{\p}_{(k_\star)} = 1 + \N{\q}_{(k_\star)} - q_1.
	\end{equation}
	From this expression follows that $\N{\p}_{(k_\star)} \leq 1$ only if either $\q = \e_1$ or if $k_\star =1$, and therefore $\p = \e_1$. These conditions are in contrast with the assumption that $\psi$ and $\varphi$ are not separable. As a consequence, equality in Eq.~\eqref{eq:conv-dist-ineq} is unachievable, which proves the second part of the theorem.
\end{proof}

Next, we show that the star conversion distance satisfies a triangle inequality. Let $\psi$, $\varphi$, and $\chi \in \pure(AB)$ and let $\p$, $\q$, and $\vec{r}$ be their Schmidt coefficients. With the help of Eq.~\eqref{eq:stardist}, this implies that
\begin{equation}\label{eq:d-star-tr-eq}
	\begin{aligned}
		T_\star(\psi \to \chi) &= \max_{n \in [d]}(\N{\p}_{(n)} - \N{\vec{r}}_{(n)}) \\
		&= \max_{n \in [d]}(\N{\p}_{(n)} - \N{\q}_{(n)} + \N{\q}_{(n)} - \N{\vec{r}}_{(n)}) \\
		&\leq \max_{n \in [d]}(\N{\p}_{(n)} - \N{\q}_{(n)}) + \max_{n \in [d]}(\N{\q}_{(n)} - \N{\vec{r}}_{(n)})  \\
		&= T_\star(\psi \to \varphi) + T_\star(\varphi \to \chi).
	\end{aligned}
\end{equation}

\emph{Remark.} The star conversion distance $T_\star(\psi \to \varphi)$ was so far only defined for pure bipartite states belonging to the same Hilbert space. This restriction is easily lifted by noting that one can always add separable auxiliary states such that the Hilbert spaces (or dimensions) match. This can be done in multiple ways: Let $\psi \in \pure(AB)$, $\varphi \in \pure(A'B')$, and let $d = \A{A} = \A{B}$, $d' = \A{A'} = \A{B'}$. Then
\begin{equation}
	T_\star(\psi_{AB} \to \varphi_{A'B'}) \equiv T_\star(\psi_{AB} \otimes \KBra{11}{11}_{A'B'} \to \KBra{11}{11}_{AB}\otimes\varphi_{A'B'}  ). 
\end{equation}
Since in fact we are only interested in the dimension of systems and their spatial separation, a more compact equivalent notation that we will use later can be defined as follows:
Denote with $m$, $m' >0$ the smallest integers such that $md = m'd'$, and with 
$A_m$, $B_m$, $A_{m'}$, $B_{m'}$ Hilbert spaces of corresponding dimensions. The conversion distance $T_\star(\psi_{AB} \to \varphi_{A'B'})$ is then
\begin{equation}\label{eq:stardist1}
	T_\star(\psi_{AB} \to \varphi_{A'B'}) \equiv T_\star(\psi_{AB} \otimes \KBra{11}{11}_{A_mB_m} \to \varphi_{A'B'} \otimes \KBra{11}{11}_{A_{m'}B_{m'}}) .
\end{equation}

\section{Universal Embezzling Families}\label{sec:emb-fam}
In the main text, we provided the following definition of universal embezzling families, which we repeat here for readability.

\begin{definition}\label{def:uef-app}
	A family of pure bipartite states $\Set{\chi_n}_{n\in \mbN}$ is called a \emph{universal embezzling family} if $\lim_{n \to \infty} T(\chi_n \to \chi_n \otimes \sigma)  = 0$ for every bipartite finite dimensional state $\sigma$. 
\end{definition}

Thanks to the triangle inequality for the star conversion distance proven in Eq.~\eqref{eq:d-star-tr-eq}, a family of states $\Set{\chi_n}_{n \in \mbN}$ is an embezzling family if and only if it can embezzle the state $\Phi_2$, as already shown in Ref~\cite[Lemma~2]{LW14} for the case of embezzling with LO.
\begin{lemma-app}[{cf. Ref.~\cite[Lemma~2 (LO)]{LW14}}]\label{lem:sufficient-embezzling}
	Let  $\Set{\chi_n}_{n \in \mbN}$ be a family of pure bipartite states. The following three statements are equivalent
	\begin{enumerate}
		\item $\Set{\chi_n}_{n \in \mbN}$ is a universal embezzling family.
		\item $\lim_{n \to \infty}T_\star(\chi_n \to \chi_n \otimes \Phi_2) = 0$.
		\item $\lim_{n \to \infty}T_\star(\chi_n \to \chi_n \otimes \Phi_m) = 0$ for every $m\ge2$.
	\end{enumerate}
\end{lemma-app}
\begin{proof}
	Clearly 2. and 3. follow from 1. due to the definition of universal embezzling families and the topological equivalence of $T$ and $T_\star$. 
	Moreover, 3. follows from 2. because
	$\Phi_2^{\otimes \ceil{\log_2 m}} \xrightarrow\locc \Phi_m$, 
	and therefore
	\begin{equation}
		\begin{aligned}
			T_\star&(\chi_n \to \chi_n \otimes \Phi_m) \\
			&\leq T_\star(\chi_n \to \chi_n \otimes \Phi_2^{\otimes \ceil{\log_2 m}})\\
			&\leq T_\star(\chi_n \to \chi_n \otimes \Phi_2^{\otimes \ceil{\log_2 m} -1}) + T_\star(\chi_n \otimes \Phi_2^{\otimes \ceil{\log_2 m} -1} \to \chi_n \otimes \Phi_2^{\otimes \ceil{\log_2 m}})\\
			&\leq T_\star(\chi_n \to \chi_n \otimes \Phi_2) +T_\star(\chi_n \otimes \Phi_2 \to \chi_n \otimes \Phi_2^{\otimes 2})+\dots \\
			&\qquad+ T_\star(\chi_n \otimes \Phi_2^{\otimes \ceil{\log_2 m} -1} \to \chi_n \otimes \Phi_2^{\otimes \ceil{\log_2 m}})\\
			&\leq \ceil{\log_2 m} T_\star(\chi_n \to \chi_n \otimes \Phi_2).
		\end{aligned}
	\end{equation}
	By taking the limit $n \to \infty$ on both sides we obtain the desired result. To conclude, we note that 3. implies 1., since for all $\sigma^{AB}$, $\Phi_{|A|}\xrightarrow\locc\sigma^{AB}$.
\end{proof}

Next, we derive a formula for the conversion distance $T_\star(\chi \to \chi \otimes \Phi_m)$ using Eq.~\eqref{eq:stardist} and Eq.~\eqref{eq:stardist1}. 

\begin{lemma-app}
	Let $\p \in \prob^\downarrow(d)$ be the Schmidt coefficients of $\chi$. Then
	\begin{equation}\label{eq:closed-star-conv-dist}
		T_\star(\chi \to \chi \otimes \Phi_m) = \max_{k \in [\SR(\chi)]}\left\{\N{\p}_{(k)} - \N{\p}_{(a_k)} - \frac{b_k}{m} p_{ a_k+1}\right\},
	\end{equation}
	where $a_k=\floor{k/m}$, $\floor{\cdot}$ denotes the floor, and $b_k=k-a_km$.
\end{lemma-app}
\begin{proof}
	The Schmidt coefficients of the input and target state are
	\begin{equation}
		\begin{aligned}
			&\p \otimes \e_1 = (p_1, \dots, p_{d}, \underbrace{0, \dots, 0}_{d \cdot (m -1) \, \text{times}}),\\
			&\p \otimes \vec{u}^{(m)} = \frac{1}{m}(\underbrace{p_{1}, \dots, p_{1}}_{m\, \text{times}},\dots, \underbrace{p_{d},\dots, p_{d}}_{m\, \text{times}}),
		\end{aligned}
	\end{equation}
	where $\vec{u}^{(m)}=(1/m,\dots,1/m)$.
	It is straightforward to see that
	\begin{equation}
		\begin{aligned}
			&\N{\p \otimes \e_1}_{(k)} = \begin{cases}
				\N{\p}_{(k)} & \text{if $k \in [d]$},\\
				1 & \text{if $d < k \leq d \cdot m$},
			\end{cases}
		\end{aligned}
	\end{equation}
	and by writing $k=a_km+b_k$, where $a_k=\floor{k/m}$,
	\begin{equation}
		\N{\p \otimes \vec{u}^{(m)}}_{(k)} =\sum_{x = 1}^{a_k} p_{x} + b_k        \frac{p_{a_k+1}}{m} =\N{\p}_{(a_k)} +  b_k\frac{p_{a_k+1}}{m}.
	\end{equation}
	Using the closed formula for the star conversion distance given in Theorem~\ref{th:closeddist-app}, we obtain
	\begin{equation}
		T_\star(\chi \to \chi \otimes \Phi_m) = \max_{k \in [\SR(\chi)]}\left\{\N{\p}_{(k)} - \N{\p}_{(a_k)} - \frac{b_k}{m} p^{(n)}_{ a_k+1}\right\}.
	\end{equation}
\end{proof}

A simplified version of this expression can be used to characterize embezzling families.
\begin{theorem}
	A family of pure bipartite states $\Set{\chi_n}_{n\in \mbN}$ with corresponding Schmidt coefficients $\Set{\p^{(n)}}_{n \in\mbN}$ is a universal embezzling family if and only if        
	\begin{equation}
		\lim_{n \to \infty} \max_{l \in [A_n]}\Set{\N{\p^{(n)}}_{{(2l -1)}} - \N{\p^{(n)}}_{{(l-1)}} } = 0,
	\end{equation}
	where $A_n = \ceil{\SR(\chi_n)/2}$.
\end{theorem}

\begin{proof}
	Due to Lemma~\ref{lem:sufficient-embezzling} and Eq.~\eqref{eq:closed-star-conv-dist}, $\Set{\chi_n}_{n\in \mbN}$ is a universal embezzling family if and only if
	\begin{equation}
		0=\lim_{n\to \infty} T_\star(\chi_n \to \chi_n \otimes \Phi_2) = \lim_{n\to \infty} \max_{k \in [\SR(\chi_n)]}\left\{\N{\p^{(n)}}_{(k)} - \N{\p^{(n)}}_{(a_k)} - \frac{b_k}{2} p^{(n)}_{ a_k+1}\right\},
	\end{equation}
	where $a_k=\floor{k/2}$ and $k=2a_k+b_k$. 
	First, we prove the necessary condition. Let $\Set{\chi_n}_{n \in \mbN}$ be a universal embezzling family. We observe that
	\begin{equation}
		\begin{aligned}
			\max_{k \in [\SR(\chi_n)]}\left\{\N{\p^{(n)}}_{(k)} - \N{\p^{(n)}}_{(a_k)} - \frac{b_k}{2} p^{(n)}_{ a_k+1}\right\} &\geq 
			\N{\p^{(n)}}_{(1)} - \N{\p^{(n)}}_{(a_1)} - \frac{b_1}{2} p^{(n)}_{ a_1+1}\\
			& = p_1^{(n)}-0-\frac{1}{2}p_1^{(n)} > 0.
		\end{aligned}
	\end{equation}
	Taking the limit for $n \to \infty$ in the expression above, we obtain $\lim_{n \to \infty} p_{ 1}^{(n)} = 0$.  Since $p_{ 1}^{(n)}$ is the largest Schmidt coefficient,  $p_{a_k +1}^{(n)}$ converges to zero too. Since $0\le\frac{b_k}{2}\le\frac{1}{2}$, 
	\begin{equation}
		\lim_{n\to\infty}  \frac{b_k}{2} p^{(n)}_{ a_k+1}=0
	\end{equation}
	and thus
	\begin{equation}
		\begin{aligned}
			0&=\lim_{n\to\infty}\max_{k \in [\SR(\chi_n)]}\left\{\N{\p^{(n)}}_{(k)} - \N{\p^{(n)}}_{(a_k)} - \frac{b_k}{2} p^{(n)}_{ a_k+1}\right\} \\
			&=\lim_{n\to\infty}\max_{k \in [\SR(\chi_n)]}\left\{\N{\p^{(n)}}_{(k)} - \N{\p^{(n)}}_{(a_k)} \right\} \\
			&\geq \lim_{n \to \infty}\max_{l \in [\ceil{\SR(\chi_n)/2}]}\Set{\N{\p^{(n)}}_{(2l-1)} - \N{\p^{(n)}}_{(l-1)}} \geq 0 .
		\end{aligned}
	\end{equation}
	
	For the sufficient condition, we observe that 
	\begin{equation}
		0 = \lim_{n \to \infty} \max_{l \in [A_n]}\Set{\N{\p^{(n)}}_{{(2l -1)}} - \N{\p^{(n)}}_{{(l-1)}} }  \geq \lim_{n \to \infty} p_1^{(n)} ,
	\end{equation}
	which implies that $\lim_{n \to \infty} p_1^{(n)} =0$. Furthermore,
	\begin{equation}
		\begin{aligned}
			\lim_{n \to \infty}T_\star(\chi_n \to \chi_n \otimes \Phi_2)  &= \lim_{n \to \infty} \max_{k \in [\SR(\chi_n)]}\left\{\N{\p^{(n)}}_{(k)} - \N{\p^{(n)}}_{(a_k)} - \frac{b_k}{2} p^{(n)}_{a_k+1}\right\} \\
			&\leq \lim_{n \to \infty}\max_{k \in [\SR(\chi_n)]}\left\{\N{\p^{(n)}}_{(k)} - \N{\p^{(n)}}_{(a_k)}\right\} 
		\end{aligned}
	\end{equation}
    At this point, we have a closer look at 
    \begin{align}
        \max_{k \in [\SR(\chi_n)]}\left\{\N{\p^{(n)}}_{(k)} - \N{\p^{(n)}}_{(a_k)}\right\}.
    \end{align}
    If $k$ is even, then $a_k=k/2=a_{k+1}$, and thus
    \begin{equation}
        \N{\p^{(n)}}_{(k)} - \N{\p^{(n)}}_{(a_k)}\le\N{\p^{(n)}}_{(k+1)} - \N{\p^{(n)}}_{(a_{k+1})}
    \end{equation}
    If $\SR(\chi_n)$ is odd, we can thus without loss of generality restrict the maximization to run over odd $k\in[\SR(\chi_n)]$. If $\SR(\chi_n)$ is even, we must additionally consider $k=\SR(\chi_n)$. Assume that this is the case: It then holds that
    \begin{equation}
		\begin{aligned}
			&\A{\N{\p^{(n)}}_{(\SR(\chi_n))} - \N{\p^{(n)}}_{\left(\SR(\chi_n)/2\right)} - \left(\N{\p^{(n)}}_{(\SR(\chi_n)-1)} - \N{\p^{(n)}}_{\left(\frac{\SR(\chi_n)}{2}-1\right)}\right)} \\
			&\qquad = \A{p^{(n)}_{\SR(\chi_n)} - p^{(n)}_{\SR(\chi_n)/2}} \\
			& \qquad \leq  p^{(n)}_1,
		\end{aligned} 
	\end{equation}
    which vanishes in the limit $n \to \infty$. As a consequence,
	\begin{equation}\label{eq:reduced-conv-dist-2}
		\begin{aligned}
			\lim_{n \to \infty}T_\star(\chi_n \to \chi_n \otimes \Phi_2) &\leq \lim_{n \to \infty}\max_{k \in [\SR(\chi_n)]}\left\{\N{\p^{(n)}}_{(k)} - \N{\p^{(n)}}_{(a_k)}\right\} \\
            &=  \lim_{n \to \infty}\max_{k \in [\SR(\chi_n)], k \text{ odd}}\left\{\N{\p^{(n)}}_{(k)} - \N{\p^{(n)}}_{(a_k)}\right\} \\
			&= \lim_{n \to \infty} \max_{l \in \ceil{\SR(\chi_n)/2}}\Set{\N{\p^{(n)}}_{(2l -1)} - \N{\p^{(n)}}_{(l-1)}} =0.
		\end{aligned}
	\end{equation}
    
	Eq.~\eqref{eq:reduced-conv-dist-2} shows that $\Set{\chi_n}_{n \in \mbN}$ is an embezzling family and concludes the proof.
\end{proof}

An important and easy to check necessary condition for a universal embezzling family is given in the following Corollary.
\begin{corollary}[{cf. Ref.~\cite[Lemma~3 (LO)]{LW14}}]\label{co:emb-nec-cond}
	If a family of pure bipartite states $\Set{\chi_n}_{n \in \mbN}$ is a universal embezzling family, then $\lim_{n \to \infty} p_1^{(n)} =0$, where $\p^{(n)} \in \prob^\downarrow (d_n)$ are the Schmidt coefficients of $\chi_n$.
\end{corollary}

\section{Regular Embezzling Families of States}\label{sec:reg-fam-states}
Often, families of states are defined in terms of a positive, monotonic, and continuous function. This motivates the following definition (compare to Ref.~\cite{LW14}).
\begin{definition-app}\label{def:reg-fam}
	A family of states $\Set{\chi_n}_{n \in \mbN}$ is a regular family if there exists a monotonic function $f\colon \mbN \to (0, \infty)$ such that 
	\begin{equation}
		\ket{\chi_n} = \frac{1}{\sqrt{F_n}}\sum_{x = 1}^{n}\sqrt{f(x)} \Ket{xx} \quad \text{for all} \quad n \in \mbN,
	\end{equation}
	where $F_n = \sum_{x=1}^{n} f(x)$.
\end{definition-app}
\emph{Remark.} If $\Set{\chi_n}_{n \in \mbN}$ is a regular family, then there exists a sequence $\Set{n_j}_{j \in \mbN}$ such that the family $\Set{\chi_{n_j}}_{j \in \mbN}$ is a universal embezzling family if and only if $\liminf_{n \to \infty} T_{\star}(\chi_n \to \chi_n \otimes \Phi_m) = 0$. However, according to our definition, the family $\Set{\chi_{n_j}}_{j \in \mbN}$ is then no longer regular, unless $n_j = j$ for all $j \in \mbN$.

For the following proofs, it is beneficial to extend $f$ to a monotonic continuous function on  $[1,\infty)$. This can often be done trivially by simply extending the domain of $f$. An example is the van Dam and Hayden family where $f(x) = x^{-1}$. Otherwise, we can extend $f$ by connecting two consecutive points with straight lines.
If the function $f$ is multiplied by a constant factor, this does not change the corresponding family of states. Therefore, from now on we assume for simplicity that $f(1) =1$.

Let $\Set{\chi_n}_{n \in \mbN}$ be a regular family of states and let $f$ be the function associated to it. We define the non-increasing functions $f_n$, $n \in \mbN$, by
\begin{equation}\label{eq:mon-f}
	f_n(x) = \begin{cases}
		f(x)& \text{if $f$ is non-increasing,}\\
		\frac{f(n+1 -x)}{f(n)} &\text{if $f$ is non-decreasing.}
	\end{cases}
\end{equation}
By our extension of $f$,  $f_n$ is also naturally extended to a continuous function $g(x, y)$ such that $g(x,n) = f_n(x)$ via the definition
\begin{equation} \label{eq:g-fun}
	g(x, y) = \begin{cases}
		f(x) & \text{if $f$ is non-increasing,}\\
		\frac{f(y+1 -x)}{f(y)} &\text{if $f$ is non-decreasing.}
	\end{cases}
\end{equation} 
The (by definition non-increasing) Schmidt coefficients of $\chi_n$ are therefore given by
\begin{equation}\label{eq:p-vect}
	p^{(n)}_{ x} = \frac{f_n(x)}{F_n}.
\end{equation}
Since $p_1^{(n)} = \frac{1}{F_n}$, we can restate Corollary~\ref{co:emb-nec-cond} for regular families.
\begin{corollary-app}[{cf. Ref.~\cite[Lemma~3 (LO)]{LW14}}] \label{co:Fn-div}
	If $\Set{\chi_n}_{n \in \mbN}$ is a regular universal embezzling family, then $\lim_{n \to \infty} F_n =+\infty$, where $F_n$ is defined in  Definition~\ref{def:reg-fam}.
\end{corollary-app}

In the following proposition, we present bounds on the limit of the conversion distance $T_\star(\chi_n \to \chi_n \otimes \Phi_m)$ in terms of the function $g(x, y)$ introduced in Eq.~\eqref{eq:g-fun}, which will be of use later.
\begin{proposition-app}\label{prop:sum-to-int}
	Let $\Set{\chi_n}_{n \in \mbN}$ be a regular family of states, $f$ be the function associated to it, and $g$ be defined as in Eq.~\eqref{eq:g-fun}. If $\lim_{n \to \infty}F_n = \infty$, then
	\begin{align}\label{eq:sum-to-int}
			&\liminf_{n \to \infty}T_\star(\chi_n \to \chi_n \otimes \Phi_m) \geq \liminf_{y \to \infty}\frac{\max_{1 \leq a\leq y/m}\Set{\int_{a}^{am} g(x, y)\diff x}}{\int_{1}^{y} g(x, y) \diff x}, \\
			\label{eq:sum-to-int-sup}
			&\limsup_{n \to \infty}T_\star(\chi_n \to \chi_n \otimes \Phi_m) \leq \limsup_{y \to \infty}\frac{\max_{1 \leq a\leq y/m}\Set{\int_{a}^{am} g(x, y)\diff x}}{\int_{1}^{y} g(x, y) \diff x}.
	\end{align}
	If $ \lim_{y \to \infty}\frac{\max_{1 \leq a\leq y/m}\Set{\int_{a}^{am} g(x, y)\diff x}}{\int_{1}^{y} g(x, y) \diff x}$ exists, this implies that
	\begin{equation}\label{eq:conv-int}
		\lim_{n \to \infty}T_\star(\chi_n \to \chi_n \otimes \Phi_m) = \lim_{y \to \infty}\frac{\max_{1 \leq a\leq y/m}\Set{\int_{a}^{am} g(x, y)\diff x}}{\int_{1}^{y} g(x, y) \diff x}.
	\end{equation}
\end{proposition-app}
\begin{proof}
	For better readability, we divide the proof into steps and use the notation $G(x, y) = \int_{1}^{x} g(t, y) \diff t$.
	
	\emph{Step 1:}
	Starting from Eq.~\eqref{eq:closed-star-conv-dist}, we find
	\begin{equation}
		\begin{aligned}
			\liminf_{n \to \infty}\,&T_\star(\chi_n \to \chi_n \otimes \Phi_m)\\
			&= \liminf_{n \to \infty} \max_{k \in [n]}\left\{\N{\p^{(n)}}_{(k)} - \N{\p^{(n)}}_{(a_k)} - \frac{b_k}{m} p^{(n)}_{ a_k+1}\right\} \\
			=&\liminf_{n \to \infty} \max_{k \in [n]}\left\{\N{\p^{(n)}}_{(a_km)} + \sum_{x=a_k m+1}^{k} p_{x}^{(n)} - \N{\p^{(n)}}_{(a_k)} - \frac{b_k}{m} p^{(n)}_{ a_k+1}\right\} \\
			=&\liminf_{n \to \infty} \max_{k \in [n]}\left\{\N{\p^{(n)}}_{(a_km)}  - \N{\p^{(n)}}_{(a_k)} + \frac{\sum_{x=a_km+1}^{a_km+b_k} f_n(x)-\frac{b_k}{m} f_n(a_k+1)}{F_n} \right\} 
		\end{aligned}
	\end{equation}
	Now we observe that the last contribution in the expression above vanishes because by construction $f_n(x)\le1$ and
	\begin{equation}
		\begin{aligned}
			&0\le \frac{\sum_{x=a_km+1}^{a_km+b_k} f_n(x)}{F_n} \le \frac{m}{F_n}, \\
			&0\le\frac{\frac{b_k}{m} f_n(a_k+1)}{F_n}\le\frac{1}{F_n}.
		\end{aligned}
	\end{equation}
	Moreover, by assumption, the right-hand sides converge to zero in the limit $n \to \infty$, and this ensures that the limit inferior is additive. As a result,
	\begin{equation}
		\begin{aligned}\label{eq:step1}
				\liminf_{n \to \infty} T_\star(\chi_n \to \chi_n \otimes \Phi_m) &= \liminf_{n \to \infty} \max_{a \in  \Set{0, \dots, \floor{n/m}}} \Set{\N{\p^{(n)}}_{(am)} - \N{\p^{(n)}}_{(a)}} \\
				&=\liminf_{n \to \infty} \max_{a \in [\floor{n/m}]} \Set{\N{\p^{(n)}}_{(am)} - \N{\p^{(n)}}_{(a)}},
		\end{aligned}
	\end{equation}
	where we used that $\N{\p^{(n)}}_{(0)} = 0$.
	
	\emph{Step 2:} 
	We start with finding bounds for $\sum_{x = 1}^n f_n(x) - \int_1^{n} f_n(x) \diff x$. To this end we observe that since $f_n(x)$ is by construction non-increasing,
	\begin{equation}\label{eq:step2}
		\int_1^{n} f_n(x) \diff x\le\int_1^{n} f_n(x) \diff x + f_n(n) \leq \sum_{x = 1}^n f_n(x) \leq f_n(1) + \int_1^{n} f_n(x)\diff x.
	\end{equation}
	By subtracting $\int_1^{n} f_n(x) \diff x$ from all terms, the desired bounds follow:
	\begin{equation}
		0 \leq \sum_{x = 1}^n f_n(x) - \int_1^{n} f_n(x) \diff x \leq f_n(1)=1.
	\end{equation}
	Dividing by $F_n=\sum_{x = 1}^n f_n(x)$ and taking the limit $n \to \infty$, we obtain $\lim_{n \to \infty} 1 - \frac{\int_{1}^{n}f_n(x) \diff x}{\sum_{x = 1}^n f_n(x)} = 0$. Since $f_n(x) = g(x,n)$, the limit can be rewritten as $\lim_{n \to \infty} \frac{G(n, n)}{F_n} = 1$. 
	
	\emph{Step 3:} 
	In this step we want to show that
	\begin{equation}
		\begin{aligned}
			\liminf_{n \to \infty} &\max_{a \in [\floor{n/m}]} \Set{\N{\p^{(n)}}_{(am)} - \N{\p^{(n)}}_{(a)}}\\
			&=\liminf_{n \to \infty} \frac{\max_{a \in [\floor{n/m}]}\Set{\sum_{x = a+1}^{am} f_n(x)}}{F_n} \\
			&=\liminf_{n \to \infty}\frac{\max_{a \in [\floor{n/m}]}\Set{\int_{a +1}^{am} f_n(x) \diff x}}{\int_1^n f_n(x)\diff x}.
		\end{aligned}
	\end{equation}
	Analogously to the previous step, one can obtain the following bounds for $\frac{\sum_{x = a+1}^{am} f_n(x)}{F_n}$:
	\begin{equation}
		\frac{\int_{a +1}^{am} f_n(x) \diff x}{F_n} \leq \frac{\sum_{x = a+1}^{am} f_n(x)}{F_n} \leq \frac{1 + \int_{a +1}^{am} f_n(x) \diff x}{F_n},
	\end{equation}
	and after taking the maximum over $a \in [\floor{n/m}]$ and the limit inferior the bounds become
	\begin{equation}\label{eq:step2.1}
		\begin{aligned}
			\liminf_{n \to \infty} &\frac{\max_{a \in [\floor{n/m}]}\Set{\int_{a +1}^{am} f_n(x) \diff x}}{F_n} \\
			&\leq \liminf_{n \to \infty}\frac{\max_{a \in [\floor{n/m}]}\Set{\sum_{x = a+1}^{am} f_n(x)}}{F_n}\\
			&\leq \liminf_{n \to \infty}\frac{1+\max_{a \in [\floor{n/m}]}\Set{\int_{a +1}^{am} f_n(x) \diff x}}{F_n}.
		\end{aligned}
	\end{equation}
	We can rewrite the last expression as
	\begin{equation}\label{eq:step2.2}
		\begin{aligned}
			\liminf_{n \to \infty}&\left(\frac{1}{F_n}+\frac{\max_{a \in [\floor{n/m}]}\Set{\int_{a +1}^{am} f_n(x) \diff x}}{F_n}\right) \\
			&=\liminf_{n \to \infty}\left(\frac{1}{F_n}\right) +\liminf_{n \to \infty}\left(\frac{\max_{a \in [\floor{n/m}]}\Set{\int_{a +1}^{am} f_n(x) \diff x}}{F_n}\right) \\
			&=\liminf_{n \to \infty}\frac{\max_{a \in [\floor{n/m}]}\Set{\int_{a +1}^{am} f_n(x) \diff x}}{F_n}.
		\end{aligned}
	\end{equation}
	By combining Eq.~\eqref{eq:step2.1} with Eq.~\eqref{eq:step2.2}, one obtains
	\begin{equation}
		\liminf_{n \to \infty} \frac{\max_{a \in [\floor{n/m}]}\Set{\sum_{x = a+1}^{am} f_n(x)}}{F_n}
		=\liminf_{n \to \infty}\frac{\max_{a \in [\floor{n/m}]}\Set{\int_{a +1}^{am} f_n(x) \diff x}}{F_n} .
	\end{equation}
	This is equivalent to
	\begin{equation}\label{eq:step3}
		\begin{aligned}
			\liminf_{n \to \infty} &\max_{a \in [\floor{n/m}]} \Set{\N{\p^{(n)}}_{(am)} - \N{\p^{(n)}}_{(a)}}\\
			&= \liminf_{n \to \infty}\frac{\max_{a \in [\floor{n/m}]}\Set{ G(am, n) - G(a+1, n)}}{G(n, n)},
		\end{aligned}
	\end{equation}
	where we replaced $F_n$ with $G(n,n)$, which we can do according to step 2.
	
	\emph{Step 4:}
	Let $c_n$ be the value that maximizes $\max_{1 \leq a\leq n/m}\Set{ G(am, n) - G(a+1, n)}$ and let $c_n'$ be the natural number satisfying $c_n'\leq c_n < c_n' +1$. This implies
	\begin{equation}
		\begin{aligned}
			G(c_nm, n) - G(c_n+1, n) &= \int_{c_n +1}^{c_nm} g(t, n) \diff t \\
			&\leq \int_{c_n' + 1}^{c_n'm}  g(t, n) \diff t + \int_{c_n'm}^{(c_n' +1)m}  g(t, n) \diff t \\
			&\leq \int_{c_n' + 1}^{c_n'm}  g(t, n)\diff t + m.
		\end{aligned}
	\end{equation}
	Dividing by $G(n, n)$, taking the limit inferior, and arguing as in the previous step, we obtain
	\begin{equation}
		\begin{aligned}
			\liminf_{n \to \infty}&\frac{\max_{1 \leq a\leq n/m}\Set{ G(am, n) - G(a+1, n)}}{G(n,n)} \\
			&= \liminf_{n \to \infty} \frac{G(c_nm, n) - G(c_n+1, n)}{G(n,n)} \\
			&\leq \liminf_{n \to \infty} \frac{G(c_n'm , n) - G(c_n' +1, n)+m }{G(n,n)}\\
			&= \liminf_{n \to \infty} \frac{G(c_n'm , n) - G(c_n' +1, n)}{G(n,n)}\\
			&\leq \liminf_{n \to \infty}\frac{\max_{a \in [\floor{n/m}]}\Set{ G(am, n) - G(a+1, n)}}{G(n,n)}.
		\end{aligned}
	\end{equation}
	The reverse inequality follows from the fact that we maximize over a larger set. This proves that
	\begin{equation}
		\label{eq:step4}
		\begin{aligned}
			\liminf_{n \to \infty}&\frac{\max_{a \in [\floor{n/m}]}\Set{ G(am, n) - G(a+1, n)}}{G(n, n)}\\
			&= \liminf_{n \to \infty}\frac{\max_{1 \leq a\leq n/m}\Set{ G(am, n) - G(a+1, n)}}{G(n, n)}.
		\end{aligned}
	\end{equation} 
	
	\emph{Step 5:} Combining the results of the previous steps, and in particular Eq.~\eqref{eq:step1}, Eq.~\eqref{eq:step3}, and Eq.~\eqref{eq:step4}, we have proven that
	\begin{equation}
		\liminf_{n \to \infty}T_\star(\chi_n \to \chi_n \otimes \Phi_m) =  \liminf_{n \to \infty}\frac{\max_{1 \leq a\leq n/m}\Set{\int_{a+1}^{am} g(x, n)\diff x}}{\int_{1}^{n} g(x, n) \diff x}.
	\end{equation}        
	In the above equation, we can replace $\int_{a+1}^{am} g(x,y) \diff x$ with $\int_{a}^{am}g(x,y)\diff x$ because the difference of the two integrals is finite and  divided by a diverging term. Thus
	\begin{equation}\label{eq:fist-step-5}
		\liminf_{n \to \infty}T_\star(\chi_n \to \chi_n \otimes \Phi_m) = \liminf_{n \to \infty}\frac{\max_{1 \leq a\leq n/m}\Set{\int_{a}^{am} g(x, y)\diff x}}{\int_{1}^{n} g(x, y) \diff x}.
	\end{equation}
	
	The expression inside the $\liminf$ on the right-hand side of Eq.~\eqref{eq:fist-step-5} is a function of $n \in \mbN\setminus \Set{ 1}$ and we call it $M(n)$,
	\begin{equation}
		M(n) = \frac{\max_{1 \leq a\leq n/m}\Set{\int_{a}^{am} g(x, n)\diff x}}{\int_{1}^{n} g(x, n) \diff x} .
	\end{equation}
	We extend this function to real numbers $y > 1$ as follows:
	\begin{equation}
		M(y) = \frac{\max_{1 \leq a\leq y/m}\Set{\int_{a}^{am} g(x, y)\diff x}}{\int_{1}^{y} g(x, y) \diff x} .
	\end{equation}        
	Since the natural numbers are a subset of the real numbers, 
	\begin{equation}\label{eq:lim-inf-conv-dist-ineq}
		\liminf_{n \to \infty} M(n) \geq \liminf_{y \to \infty} M(y) ,
	\end{equation}
	which proves Eq.~\eqref{eq:sum-to-int}. 
	Furthermore, Steps 1-4 can be repeated with exactly the same arguments for the $\limsup$, and they imply
	\begin{equation}\label{eq:lim-sup-conv-dist-ineq}
		\limsup_{n \to \infty}T_\star(\chi_n \to \chi_n \otimes \Phi_m) =  \limsup_{n \to \infty}M(n) \le \limsup_{y \to \infty} M(y) .
	\end{equation}
	This proves Eq.~\eqref{eq:sum-to-int-sup}.
	
	In the following, we assume that $\lim_{y \to \infty}M(y)$ exists. This implies
	\begin{equation} \label{eq:lim-equals}
		\lim_{y \to \infty} M(y) = \liminf_{y \to \infty} M(y) = \limsup_{y \to \infty} M(y) .
	\end{equation}
	From Eq.~\eqref{eq:lim-inf-conv-dist-ineq}, Eq.~\eqref{eq:lim-sup-conv-dist-ineq}, and Eq.~\eqref{eq:lim-equals} we obtain
	\begin{equation}\label{eq:lim-inf-conv-dist}
		\lim_{y \to \infty} M(y) \leq \liminf_{n \to \infty} M(n) \leq \limsup_{n \to \infty} M(n) \leq \lim_{y \to \infty} M(y).
	\end{equation}
	Therefore, $\lim_{n \to \infty} M(n)$ exists and is equal to $\lim_{y \to \infty} M(y)$. From Eq.~\eqref{eq:fist-step-5} and Eq.~\eqref{eq:lim-sup-conv-dist-ineq} we derive
	\begin{equation}
		\lim_{n \to \infty}T_\star(\chi_n \to \chi_n \otimes \Phi_m) = \lim_{n \to \infty} M(n) = \lim_{y \to \infty}M(y) .            
	\end{equation}
	This proves Eq.~\eqref{eq:conv-int} and concludes the proof.

\end{proof}

Corollary~\ref{co:Fn-div} combined with Proposition~\ref{prop:sum-to-int} leads to the following characterization of embezzling families.
\begin{corollary-app}
	Let $\Set{\chi_n}_{n \in \mbN}$ be a regular family of states, $f$ be the function associated to it, and $g$ be defined as in Eq.~\eqref{eq:g-fun}. If 
	\begin{equation}
		\lim_{y \to \infty}\frac{\max_{1 \leq a\leq y/m}\Set{\int_{a}^{am} g(x, y)\diff x}}{\int_{1}^{y} g(x, y) \diff x}
	\end{equation} exists, then the family of states $\Set{\chi_n}_{n \in \mbN}$ is an universal embezzling family if and only if $\lim_{n \to \infty }F_n= \infty$ and 
	\begin{equation}
		\lim_{y \to \infty}\frac{\max_{1 \leq a\leq y/m}\Set{\int_{a}^{am} g(x, y)\diff x}}{\int_{1}^{y} g(x, y) \diff x} = 0.
	\end{equation}
\end{corollary-app}

In addition, thanks to Proposition~\ref{prop:sum-to-int}, the task of determining if a family of states is a universal embezzling family is converted into an optimization problem. It is enough to find the maximum of the differentiable function $G_y(a) = \int_{a}^{am} g(x, y)\diff x$ on $1\le a\le y/m$, which is easily done with the help of its derivative 
\begin{equation}\label{eq:diff}
	m g(am, y) - g(a, y)  , \quad 1 \leq a \leq y/m.
\end{equation}

\section{Generalization of the van Dam and Hayden Family}\label{sec:gen-van-dam}
The universal embezzling family $\Set{\chi_n}_{n \in \mbN}$ introduced by van Dam and Hayden  consists of the states
\begin{equation}
	\Ket{\chi_n} = \frac{1}{\sqrt{H_n}}\sum_{x=1}^{n}\sqrt{x^{-1} }\Ket{xx},
\end{equation}
where $H_n = \sum_{x=1}^{n}x^{-1}$ is the $n$-th harmonic number. We generalize this family of states as follows: For every $\alpha \in \mbR$ we introduce the family $\Set{\chi^{(\alpha)}_n}_{n \in \mbN}$, where
\begin{equation}
	\Ket{\chi^{(\alpha)}_n} = \frac{1}{\sqrt{H^{(-\alpha)}_n }}\sum_{x=1}^{n}\sqrt{x^{\alpha}} \Ket{xx}
\end{equation}
and $H^{(-\alpha)}_n = \sum_{x=1}^{n}x^{\alpha}$ is the $n$-th generalized harmonic number. Note that the families of states $\Set{\chi_n^{(\alpha)}}_{n \in \mbN}$ are regular families of states, and the van Dam and Hayden family is recovered for $\alpha = -1$. The corresponding Schmidt coefficients, as per our convention arranged in non-increasing order, are 
\begin{equation}
	\p^{(n|\alpha)} = \begin{cases}
		\frac{1}{H^{(-\alpha)}_n } (1, 2^\alpha,\dots, n^\alpha) & \text{if $\alpha \leq 0$},\\
		\frac{1}{H^{(-\alpha)}_n } (n^\alpha, \dots,2^\alpha, 1) & \text{if $\alpha > 0$}.
	\end{cases}
\end{equation}
In the remaining part of this section, we show that the family of state $\Set{\chi^{(\alpha)}_n}_{n \in  \mbN}$ is a universal embezzling family if and only if $\alpha = -1$ and derive bounds on the star conversion distance.
\paragraph*{Case $\alpha < -1$.} For $\alpha < -1$, $\lim_{n\to \infty} H^{(-\alpha)}_n $ is finite, thus the regular family of states $\Set{\chi^{\alpha}_n}_{n \in \mbN}$ is not a universal embezzling family (see Corollary~\ref{co:Fn-div}). The largest Schmidt coefficient provides a lower bound on the star conversion distance: Using Eq.~\eqref{eq:stardist}, we obtain
\begin{equation}
	\begin{aligned}
		T_\star(\chi^{(\alpha)}_n \to \chi^{(\alpha)}_n \otimes \Phi_m) &= \max_{k \in [n]} \left\{ \N{\p^{(n|\alpha)}}_{(k)} - \N{\p^{(n|\alpha)}}_{(a_k)} - \frac{b_k}{m} p^{ ( n | \alpha)}_{a_k+1} \right\} \\
		&\geq p_{1}^{(n|\alpha)} \left(1 - \frac{1}{m}\right) \\
		&= \frac{1}{H^{(-\alpha)}_n }\left(1 - \frac{1}{m}\right),
	\end{aligned}
\end{equation}
where the inequality followed from choosing $k =1$. By taking the limit, we obtain
\begin{equation}
	\liminf_{n \to \infty} T_\star(\chi^{(\alpha)}_n \to \chi^{(\alpha)}_n \otimes \Phi_m)  \geq \liminf_{n \to \infty} \frac{1}{H^{(-\alpha)}_n }\left(1 - \frac{1}{m}\right) = \frac{1}{\zeta(-\alpha)}\left(1 - \frac{1}{m}\right),
\end{equation}
where $\zeta(-\alpha) = \sum_{x = 1}^\infty x^\alpha$ is the Riemann Zeta function. To obtain an upper bound, we  observe that, since $\alpha <-1$,
\begin{equation}
	\sum_{x=a+1}^{(a+1)m} x^\alpha \leq (a+1)^\alpha + \int_{a+1}^{(a+1)m} x^\alpha \diff x = (a+1)^\alpha + \frac{(a+1)^{\alpha +1}}{\alpha +1}\left(m^{\alpha +1} - 1\right).
\end{equation}
The right-hand side is decreasing in $a$, and thus
\begin{equation}
	\begin{aligned}
		\max_{k \in [n]}&\left\{\N{\p^{(n|\alpha)}}_{{(k)}} - \N{\p^{(n|\alpha)}}_{(a_k)} - \frac{b_k}{m} p_{ a_k+1}^{(n|\alpha)}\right\} \\
		&\le \frac{1}{H^{(-\alpha)}_n } \max_{k \in [n]} \Set{\sum_{x=a_k+1}^{k} x^\alpha }\\
		&\le \frac{1}{H^{(-\alpha)}_n }\max_{a \in \{0,...,\floor{n/m}\}}\Set{\sum_{x=a+1}^{(a+1)m} x^\alpha }\\
		&\le \frac{1}{H^{(-\alpha)}_n }\max_{a \in \{0,...,\floor{n/m}\}}\Set{(a+1)^\alpha + \frac{(a+1)^{\alpha +1}}{\alpha +1}\left(m^{\alpha +1} - 1\right)} \\
		&=\frac{1}{H^{(-\alpha)}_n }\left(1 + \frac{m^{\alpha +1} - 1}{\alpha +1}\right).
	\end{aligned}
\end{equation}
Taking the limit $n \to \infty$, and considering that the conversion distance is by definition smaller than one, we obtain
\begin{equation}
	\limsup_{n \to \infty} T_\star(\chi^{(\alpha)}_n \to \chi^{(\alpha)}_n \otimes \Phi_m) \leq \min\Set{1, \frac{1}{\zeta(-\alpha)}\left(1 + \frac{m^{\alpha +1} - 1}{\alpha +1}\right)}.
\end{equation}
\paragraph*{Case $\alpha = -1$.}
This is the van Dam and Hayden family. Analogously to the previous case, we observe that
\begin{equation}
	\sum_{x=a+1}^{(a+1)m} x^{-1} \leq (a+1)^{-1} + \int_{a+1}^{(a+1)m} x^{-1} \diff x =  (a+1)^{-1} + \ln m \leq 1 + \ln m .
\end{equation}
Since $H_n^{(1)} = H_n$ diverges, by following the same steps, we have
\begin{equation}
	\lim_{n \to \infty} T_\star(\chi^{(-1)}_n \to \chi^{(-1)}_n \otimes \Phi_m) \leq \lim_{n \to \infty} \frac{1 + \ln m }{H_n} = 0.
\end{equation}
Thus, the limit for the conversion distance is zero as expected for the van Dam and Hayden family.
\paragraph*{Case $-1 < \alpha < 0$.}
In this scenario $H^{(-\alpha)}_n $ diverges when $n \to \infty$, and we can use the results of Section~\ref{sec:reg-fam-states}. Here $g(x, y) = x^\alpha$, and the derivative with respect to $a$ of the function $\int_{a}^{am} g(x, y)\diff x$ is
\begin{equation}
	m g(am, y) - g(a, y) = a^\alpha (m^{\alpha +1} - 1),
\end{equation}
which is positive in the domain $1 \leq a \leq y/m$. This implies that the function $\int_{a}^{am} g(x, y)\diff x$ is non-decreasing and the maximum 
\begin{equation}
	\max_{1 \leq a\leq y/m}\Set{\int_{a}^{am} g(x, y)\diff x}
\end{equation}
is obtained for $a = y/m$. Due to Proposition~\ref{prop:sum-to-int},
\begin{equation}
	\lim_{n \to \infty}T_\star(\chi^{(\alpha)}_n \to \chi^{(\alpha)}_n \otimes \Phi_m) = \lim_{y \to \infty}\frac{\int_{y/m}^{y} x^\alpha \diff x}{\int_{1}^{y} x^\alpha \diff x} = 1 - \frac{1}{m^{\alpha +1}}.
\end{equation}
\paragraph*{Case $\alpha = 0$.}
Here we have
\begin{equation}
	p_{x}^{(n)} = \frac{1}{n}\, , \quad (\p^{(n)} \otimes \vec{u}^{(m)})_{ x} = \frac{1}{nm}.
\end{equation}	
We can compute the star conversion distance directly using  Eq.~\eqref{eq:stardist} and obtain
\begin{equation}
	T_\star(\chi^{(0)}_n \to \chi^{(0)}_n \otimes \Phi_m) = \max_{k \in [n]} \Set{\frac{k}{n} - \frac{k}{nm}} = 1 - \frac{1}{m}.
\end{equation}
\paragraph*{Case $\alpha > 0$.} Also here $H^{(-\alpha)}_n $ diverges for $n \to \infty$ and we can use the results of Section~\ref{sec:reg-fam-states}. The function $g(x, y)$ is defined as $g(x,y) =(y +1 -x)^\alpha / y^\alpha$, because $x^\alpha$ is increasing. From this follows
\begin{equation}\label{eq:a-max}
	\begin{aligned}
		m g(am, y) - g(a, y) \geq 0 &\Leftrightarrow m(y +1 -am)^\alpha \geq (y +1 -a)^\alpha \\
		&\Leftrightarrow a \leq \frac{(y +1)\left(m^{\frac{1}{\alpha}} -1\right)}{m^{1 + \frac{1}{\alpha}}-1} \coloneqq a_{\max}.
	\end{aligned}
\end{equation}
Since $a_{\max}$ increases linearly with $y$ and $a_{\max}\le \frac{y}{m}$ is equivalent to $y\ge \frac{m\left(m^{1/\alpha}-1\right)}{m-1}$,
for large enough $y$, the value $a_{\max}$ belongs to the interval $[1, y/m]$ and is a global maximum. 
Plugging $a_{\max}$ into Proposition~\ref{prop:sum-to-int}, we obtain
\begin{equation}\label{eq:a-geq-0}
	\lim_{n \to \infty}T_\star(\chi^{(\alpha)}_n \to \chi^{(\alpha)}_n \otimes \Phi_m) = \lim_{y \to \infty}\frac{\int_{a_{\max}}^{a_{\max} m} \frac{(y +1 -x)^\alpha}{y^\alpha} \diff x}{\int_{1}^{y} \frac{(y +1 -x)^\alpha}{y^\alpha} \diff x} = (m -1)\left(\frac{m-1}{m^{1 + \frac{1}{\alpha}} -1}\right)^\alpha.
\end{equation}
This completes the study of regular families of states defined by $f(x)=x^\alpha$. The star conversion distance vanishes only for $\alpha = -1$, thus the only universal embezzling family of this form is the one introduced by van Dam and Hayden. Furthermore, exact values for the limit of the star conversion distance for $\alpha \ge-1$ and lower and upper bounds for $\alpha<-1$ were provided.
\section{Uniqueness of the van Dam and Hayden Embezzling Family}\label{sec:uniq-vdh}
In this section, we provide further results on the uniqueness of the van Dam and Hayden embezzling family. Given a regular family of states $\Set{\chi_n}_{n \in \mbN}$, the asymptotic behavior of the function $f$ associated to it is relevant to determine whether $\Set{\chi_n}_{n \in \mbN}$ is a universal embezzling family or not (e.g., see Corollary~\ref{co:Fn-div}). We thus use the notations little-$\omega$ and little-$o$ (see, e.g., Ref.~\cite{Cor22}) to describe asymptotic relations between two functions $g, h\colon [1, \infty) \to (0, \infty)$,
\begin{equation}
	h \in \omega(g) \Leftrightarrow \lim_{x \to \infty}\frac{h(x)}{g(x)} = + \infty \, , \quad 	h \in o(g) \Leftrightarrow \lim_{x \to \infty}\frac{h(x)}{g(x)} = 0 .
\end{equation}

Before we state the results about the uniqueness of the van Dam and Hayden embezzling family, we prove the following Lemma, which is a direct consequence of Theorem 2 in Ref.~\cite{Wij85}, which we restate here to improve readability: Let $\mu$ be a measure on the real line $\mbR$, and let $f_i$, $g_i$ ($i = 1,2$) be four Borel-measurable functions: $\mbR \to \mbR$ such that $f_2 \geq 0$ and $g_2 \geq 0$, and $\int \A{f_i g_j} \diff\mu < \infty$ ($i,j = 1,2$). If $f_1/f_2$ and $g_1/g_2$ are monotonic in the same direction, then
\begin{equation}
	\int f_1 g_1 \diff\mu \int f_2 g_2 \diff \mu \geq \int f_1 g_2 \diff \mu \int f_2 g_1 \diff \mu .
	\end{equation}
\begin{lemma-app}\label{le:int-ineq}
	Let $f$, $g$ be continuous, positive functions on $[a, b]$ such that $f(x)/g(x)$ is a non-decreasing function on $[a, b]$. Then, for every $x_1$ such that $a < x_1 < b$, 
	\begin{equation}
		\int_{a}^b f(x) \diff x \int_a^{x_1} g(x)\diff x \geq \int_{a}^{x_1} f(x) \diff x \int_a^{b} g(x) \diff x.
	\end{equation}
\end{lemma-app}
\begin{proof}
	Let $\tilde{h}(x) = 1$ and $h(x) = \chi_{[a, x_1]}(x)$, where $\chi_S(x)$ is the characteristic function of the set $S$. Furthermore, let
	\begin{equation}
		h_k(x) = \begin{cases}
			1 & \text{if $a\leq x \leq x_1$}, \\
			e^{-k(x - x_1)} &\text{if $x_1 < x \leq b$}.
		\end{cases}
	\end{equation}
	The sequence of functions $\Set{h_k}_{k \in \mbN}$ converges to $h(x)$ and satisfies $0\le h_k(x) \leq 1$. Due to the dominated convergence theorem (see, e.g., Ref.~\cite{Rud87}), we have
	\begin{equation}
		\lim_{k \to \infty} \int_{a}^{b} f(x)h_k(x) = \int_{a}^{x_1} f(x).
	\end{equation}
	The same result holds for $g$.
	
	Since both $f(x)/g(x)$ and $\tilde{h}(x)/h_k(x)$ are non-decreasing, we can apply Theorem 2 of Ref.~\cite{Wij85} to find
	\begin{equation}
		\int_{a}^b f(x) \tilde{h}(x) \diff x\int_{a}^{b} g(x)h_k(x)\diff x \geq \int_a^b f(x) h_k(x) \diff x \int_a^b g(x) \tilde{h}(x)\diff x.
	\end{equation}
	After taking the limit $k \to \infty$ on both sides, we obtain
	\begin{equation}
		\int_{a}^b f(x) \diff x \int_a^{x_1} g(x)\diff x \geq \int_{a}^{x_1} f(x) \diff x \int_a^{b} g(x) \diff x,
	\end{equation}
	which finishes the proof.
\end{proof}

With this lemma at hand, we are ready to present the promised results concerning the uniqueness of the van Dam and Hayden embezzling family.

\begin{theorem-app}\label{th:uniq-vDH}
	Let $f$ be a positive non-increasing function such that   $f(x)/x^\alpha$ is  asymptotically monotonic for all $\alpha \in \mbR$ and let $\Set{\chi_n}_{n \in \mbN}$ be the regular family of states associated to it (see Definition~\ref{def:reg-fam}). Then $\Set{\chi_n}_{n \in \mbN}$ is a universal embezzling family if and only if $f \in \omega(x^{-1-\varepsilon})\cap o(x^{-1 +\varepsilon})\ \forall \varepsilon>0$ and $\sum_{x = 1}^\infty f(x) = \infty$. Furthermore, if $f \notin \omega(x^{-1-\varepsilon})\cap o(x^{-1 +\varepsilon})$ for at least one $\varepsilon>0$, then  $\Set{\chi_{n_j}}_{j \in \mbN}$, where $\Set{n_j}_{j \in \mbN}$ is any sequence of natural numbers, is not a universal embezzling family.
\end{theorem-app}
\begin{proof} Before we start, we notice that if $\lim_{j \to \infty} n_j = J <\infty$, then $\lim_{j \to \infty}F_{n_j} = \sum_{x = 1}^J f(x) < \infty$, which implies that the family of states $\Set{\chi_{n_j}}_{j \in \mbN}$ is not an embezzling family (see Corollary~\ref{co:Fn-div}). In this proof we will thus assume, w.l.o.g., that $\lim_{j \to \infty} n_j = \infty$.
	
	\emph{Necessary condition  ---}
	Let $\Set{\chi_n}_{n \in \mbN}$ be a universal embezzling family with corresponding function $f$ satisfying the assumptions above. We already proved in Corollary~\ref{co:Fn-div}  that if $\Set{\chi_n}_{n \in \mbN}$ is a universal embezzling family, then $\sum_{x = 1}^\infty f(x) = + \infty$. To show the remainder, let
	\begin{equation}
		R= \Set{\alpha \in \mbR | \lim_{x \to \infty} \frac{f(x)}{x^\alpha} = 0} \, , \quad L= \Set{\alpha \in \mbR | \lim_{x \to \infty} \frac{f(x)}{x^\alpha} = +\infty}.
	\end{equation}
	Since $\lim_{x \to \infty} f(x) /x^\alpha =0$ for $\alpha >0$, we have that $R \neq \emptyset$. Now we prove by contradiction that also $L \neq \emptyset$. Let us assume that $L$ is empty. This implies that $\lim_{x \to \infty}  f(x)/x^\alpha = l_\alpha < \infty$ for all $\alpha$. If there exists an $\alpha$ such that $l_\alpha \neq 0$, then $\lim_{x \to \infty} f(x)/x^{\alpha - 1} =  l_\alpha (\lim_{x \to \infty} x) = +\infty$. Thus, $\alpha - 1 \in L$, and $L \neq \emptyset$, leading to the desired contradiction. If instead $\lim_{x \to \infty} f(x)/x^\alpha = 0$ for all $\alpha$, then $\lim_{x \to \infty} f(x) / x^{-2} = 0$. This implies that $f(x)/x^{-2}$ converges monotonically to zero  for large $x$, i.e., there exists an $N$ such that $f(x) < x^{-2}$ for $x > N$. Thus, $\sum_{x=N}^\infty f(x)  < \sum_{x=N}^\infty x^{-2}  < \infty$. According to Corollary~\ref{co:Fn-div}, this contradicts the hypothesis that the family under consideration is universally embezzling. We have therefore shown that $L \neq \emptyset \neq R$. 
	
	The next step is to prove that $\inf R = \sup L$. From the definition of $R$ and $L$ follows that $\inf R \geq \sup L$. Let us assume that $\inf R > \sup L$, i.e., that there exists an $\alpha \in \mbR$ such that $\sup L < \alpha <\inf R$. Since $\alpha \notin L \cup R$, there exists a positive real number $l$ such that $\lim_{x \to \infty} f(x)/x^\alpha = l $. 
	Now pick $\alpha_1$ such that $\sup L <\alpha_1 < \alpha <\inf R$. We observe that $\lim_{x \to \infty}f(x)/x^{\alpha_1} = l (\lim_{x\to \infty} x^\alpha / x^{\alpha_1} )= + \infty$, thus $\alpha_1 \in L$. This is in contradiction to the choice $\alpha_1 > \sup L$. We therefore showed that $\sup L = \inf R$. 
	
	So far, we have shown that if $\Set{\chi_n}$ is a universal embezzling family satisfying our assumptions, then there exists a unique $\alpha$ such that $f(x) \in o(x^{\alpha+\varepsilon})\cap \omega(x^{\alpha-\varepsilon})$ for every $\varepsilon > 0$. What is left to show is that if $\alpha \neq -1$, then the family of states corresponding to $f$ cannot be universally embezzling. From the above discussion, we know that $\Set{\alpha > 0} \subseteq R$, thus we can focus on $\alpha \leq 0$ and split our discussion into two scenarios, $\alpha < -1$ and $-1 < \alpha \leq0$.
	
	If $\alpha < -1$, then there exists an $\varepsilon >0$ such that $\alpha + \varepsilon < -1$. Since $f \in o(x^{\alpha + \varepsilon})$ and $f(x)/ x^{\alpha + \varepsilon}$ converges monotonically to zero for large $x$, there exists an $N$ such that $f(x) < x^{\alpha + \varepsilon}$ for $x > N$. This implies that $\sum_{x=N}^\infty f(x)  < \sum_{x=N}^\infty x^{\alpha + \varepsilon}  < \infty$  and therefore $\sum_{x=1}^\infty f(x) <\infty$, which, according to Corollary~\ref{co:Fn-div}, contradicts the hypothesis that $\Set{\chi_n}$ is a universal embezzling family. Let $\Set{n_j}_{j \in \mbN}$ be any sequence of natural number such that $\lim_{j \to \infty} n_j = \infty$. Then $\lim_{j \to \infty}\sum_{x = 1}^{n_j} f(x) = \sum_{x = 1}^\infty f(x) < \infty$. This implies, again due to Corollary~\ref{co:Fn-div}, that $\Set{\chi_{n_j}}_{j \in \mbN}$ is not a universal embezzling family.
	
	If $-1<\alpha\leq 0$, then there exists an $\varepsilon > 0$ such that $\alpha - \varepsilon > -1$. We notice that $f(x)/ x^{\alpha - \varepsilon}$ diverges to infinity for large $x$, because $f \in \omega(x^{\alpha - \varepsilon})$. This implies that there exists an $N$ such that $f(x) \ge x^{\alpha -\varepsilon}$ for $x > N$ and therefore $\sum_{x=N}^\infty f(x)  \ge \sum_{x=N}^\infty x^{\alpha - \varepsilon}  = \infty$,        
	which allows us to use the results of Section~\ref{sec:reg-fam-states}. Since $f(x)/x^{\alpha - \varepsilon}$ is non-decreasing for $x > N$, for all $y$ such that $y/m >N$, we can use Lemma~\ref{le:int-ineq} to obtain 
	\begin{equation}
		\int_{N}^{y} f(x) \diff x \int_{N}^{y/m} x^{\alpha - \varepsilon} \diff x \geq \int_{N}^{y/m} f(x) \diff x \int_{N}^{y} x^{\alpha - \varepsilon}\diff x.
	\end{equation}
	This implies that 
	\begin{equation}
		\frac{(y/m)^{\alpha - \varepsilon + 1} - N^{\alpha - \varepsilon + 1}}{(y)^{\alpha - \varepsilon + 1} - N^{\alpha - \varepsilon + 1}}\ge \frac{\int_{N}^{y/m} f(x) \diff x}{\int_{N}^{y} f(x) \diff x}.
	\end{equation}
	Next we introduce $\delta := \alpha - \varepsilon + 1 > 0$ and take on both sides the limit inferior $y \to \infty$, leading to 
	\begin{equation}
		\begin{aligned}
			\frac{1}{m^\delta} &= \liminf_{y \to \infty} \frac{(y/m)^{\alpha - \varepsilon + 1} - N^{\alpha - \varepsilon + 1}}{(y)^{\alpha - \varepsilon + 1} - N^{\alpha - \varepsilon + 1}} \\
			&\geq \liminf_{y \to \infty} \frac{\int_{N}^{y/m} f(x) \diff x}{\int_{N}^{y} f(x) \diff x} \\
			&= \liminf_{y \to \infty} \frac{\int_{1}^{y/m} f(x) \diff x}{\int_{1}^{y} f(x) \diff x},
		\end{aligned}
	\end{equation}
	where in the last equality, we used that $\int_{1}^{N} f(x) \diff x$ is finite, whilst  $\int_{N}^{\infty} f(x) \diff x$ diverges.
	As last step, we  observe that according to Proposition~\ref{prop:sum-to-int},
	\begin{equation}
		\begin{aligned}
			\liminf_{n \to \infty}T_\star(\chi_n \to \chi_n \otimes \Phi_m) &\ge \liminf_{y \to \infty}\frac{\max_{1 \leq a \leq y/m}\Set{\int_{a}^{am} f(x)\diff x}}{\int_{1}^{y}f(x) \diff x}\\
			&\geq \liminf_{y \to \infty}\frac{\int_{y/m}^{y} f(x)\diff x}{\int_{1}^{y} f(x) \diff x} \\
			&\geq 1 - \frac{1}{m^\delta} \\
			&> 0.
		\end{aligned}
	\end{equation}
	
	Using the remark after Definition~\ref{def:reg-fam}, we obtain that for $-1 <\alpha \leq 0$, there are no sequences $\Set{n_j}_{j \in \mbN}$ such that $\Set{\chi_{n_j}}_{j \in \mbN}$ is a universal embezzling family, contradicting the hypothesis.
	Once we combine the results for $\alpha < -1$ and $-1< \alpha \leq 0$, we have that if $\Set{\chi_n}$ is a universal embezzling family satisfying our assumptions, $f(x) \in o(x^{-1 + \varepsilon}) \cap \omega(x^{-1 -\varepsilon})$ for all $\varepsilon >0$. We have also shown that if $f \notin o(x^{-1 + \varepsilon}) \cap \omega(x^{-1 -\varepsilon})$ for at least one $\varepsilon >0$, then there are no sequences $\Set{n_j}_{j \in \mbN}$ such that the family $\Set{\chi_{n_j}}_{j \in \mbN}$ is a universal embezzling family. This concludes the first part of the proof.
	
	\emph{Sufficient condition ---} Let $f$ be a function satysfying our assumptions. Since $f(x)/x^{-1}$ is asymptotically monotonic by assumption,  $\lim_{x \to \infty} f(x)/x^{-1}$ exists in $[0, +\infty]$. A priori, it can be either $0$, $0<l \in \mbR$, or $+ \infty$. 		
	If $\lim_{x \to \infty} f(x)/x^{-1} = l >0$, then by definition of the limit, for any $l>\tilde\varepsilon > 0$ there exists an $N$ such that $\A{\frac{f(x)}{x^{-1}} - l} <\tilde\varepsilon$ for all $x >N$. This is equivalent to $(l-\tilde\varepsilon) x^{-1} < f(x) < (l + \tilde\varepsilon) x^{-1}$ for all $x > N$.  Since $\sum_{x =N}^{\infty}f(x) \geq (l - \tilde\varepsilon)\sum_{x = N}^\infty x^{-1} = +\infty$, we can compute the star conversion distance using Proposition~\ref{prop:sum-to-int}, and find
	\begin{equation}\label{eq:bound-dist}
		\begin{aligned}
			\lim_{n \to \infty}&T_\star(\chi_n \to \chi_n \otimes \Phi_m) \\
			&= \lim_{y \to \infty}\frac{\max_{1 \leq a \leq y/m}\Set{\int_{a}^{am} f(x)\diff x}}{\int_{1}^{y}f(x) \diff x} \\
			& \leq \lim_{y \to \infty}\frac{\max_{1 \leq a < N}\Set{\int_{a}^{am} f(x)\diff x} + \max_{N\leq a \leq y/m}\Set{\int_{a}^{am} f(x)\diff x}}{\int_{1}^{y}f(x) \diff x}\\
			&= \lim_{y \to \infty}\frac{\max_{N\leq a \leq y/m}\Set{\int_{a}^{am} f(x)\diff x}}{\int_{1}^{y}f(x) \diff x} \\
			&\leq \lim_{y \to \infty}\frac{\max_{N\leq a \leq y/m}\Set{\int_{a}^{am} f(x)\diff x}}{\int_{N}^{y}f(x) \diff x} \\
			& \leq \frac{l + \tilde\varepsilon}{l - \tilde\varepsilon} \lim_{y \to \infty}\frac{\max_{N\leq a \leq y/m}\Set{\int_{a}^{am} x^{-1}\diff x}}{\int_{N}^{y}x^{-1} \diff x} \\
			&= \frac{l + \tilde\varepsilon}{l - \tilde\varepsilon}  \lim_{y \to \infty} \frac{\log m}{\log y - \log N}\\
			&= 0.
		\end{aligned}
	\end{equation}
	Thus if $\lim_{x \to \infty} f(x)/x^{-1} = l \neq 0$, then $\Set{\chi_n}$ is a universal embezzling family.
	
	Suppose now $\lim_{x \to \infty} f(x)/x^{-1} = 0$. Since $f(x)/x^{-1}$ is non-increasing for large $x$, there exists an $\tilde{N}$ such that $f(x) < x^{-1}$ for all $x > \tilde{N}$. For any fixed $\varepsilon >0$, by hypothesis, $\lim_{x \to \infty} f(x)/ x^{-1 -\varepsilon} = + \infty$. Furthermore, since by assumption $\int_1^\infty f(x) \diff x= \infty$, while $\int_{1}^{\infty} x^{-1 -\varepsilon} \diff x < \infty$, there exists an $\tilde{N}_\varepsilon$, such that $\int_1^y f(x) \diff x \geq \int_{1}^{y} x^{-1 -\varepsilon} \diff x$ for $y > \tilde{N}_\varepsilon$. Let $N_\varepsilon = \max\Set{\tilde{N}, \tilde{N}_\varepsilon}$. Using again Proposition~\ref{prop:sum-to-int} and performing the same steps as in Eq.~\eqref{eq:bound-dist}, we obtain 
	\begin{equation}
		\begin{aligned}
			\limsup_{n \to \infty}T_\star(\chi_n \to \chi_n \otimes \Phi_m) &\leq \limsup_{y \to \infty}\frac{\max_{N_\varepsilon\leq a \leq y/m}\Set{\int_{a}^{am} f(x)\diff x}}{\int_{1}^{y}f(x) \diff x} \\
			&\leq \limsup_{y \to \infty}\frac{\max_{N_\varepsilon\leq a \leq y/m}\Set{\int_{a}^{am} x^{-1}\diff x}}{\int_{1}^{y}x^{-1 - \varepsilon} \diff x} \\
			&= \log m \lim_{y \to \infty}\frac{\varepsilon}{1 - y^{-\varepsilon}} \\
			&= \varepsilon\log m.
		\end{aligned}
	\end{equation}
	Since this is true for all $\varepsilon > 0$,
	\begin{equation}
		\limsup_{n \to \infty}T_\star(\chi_n \to \chi_n \otimes \Phi_m) \leq \lim_{\varepsilon \to 0} \varepsilon \log m = 0.
	\end{equation}
	This implies that $\lim_{n \to \infty}T_\star(\chi_n \to \chi_n \otimes \Phi_m)=0$, i.e., $\{\chi_n\}$ is a universal embezzling family.
	
	Let us now consider the case $\lim_{x \to \infty} f(x)/x^{-1} = \infty$. From this immediately follows that $\sum_{x = 1}^\infty f(x) = +\infty$, so we can use the results of Section~\ref{sec:reg-fam-states}. Furthermore, there exists an $N$ such that $f(x)/x^{-1}$ is non-decreasing for $x \ge N$. We write $f(x) = x^{-1}h(x)$ and thus $h(x)$ is non-decreasing for $x \ge N$. Computing the derivative of $\int_a^{am} f(x) \diff x$ for $a \ge N$, we find
	\begin{equation}
		\deriv{\int_{a}^{am} f(x)\diff x}{a} = m f(am) - f(a) =\frac{h(am) - h(a)}{a} \geq 0.
	\end{equation}
	Thus, $\max_{N\le a \le y/m} \int_{a}^{am} f(x)\diff x = \int_{y/m}^{y} f(x) \diff x$ and, by following the same steps as in Eq.~\eqref{eq:bound-dist}, we obtain
	\begin{equation}\label{eq:star-dist-up-bound}
		\begin{aligned}
			\limsup_{n \to \infty}T_\star(\chi_n \to \chi_n \otimes \Phi_m) &\leq \limsup_{y \to \infty}\frac{\max_{N\le a \le y/m}\Set{\int_{a}^{am} f(x)\diff x}}{\int_{1}^{y}f(x) \diff x} \\
			&= \limsup_{y \to \infty}\frac{\int_{y/m}^{y} f(x)\diff x}{\int_{1}^{y}f(x) \diff x}.
		\end{aligned}
	\end{equation}
	Let us fix $\varepsilon >0$. Since $f(x)/x^{-1 + \varepsilon}$ is non-increasing for $x > N_\varepsilon$ by hypothesis, $x^{-1 + \varepsilon}/f(x)$ is non-decreasing for $x > N_\varepsilon$. By applying Lemma~\ref{le:int-ineq} (for $y$ large enough), we obtain
	\begin{equation}
		\int_{N_\varepsilon}^y x^{-1 +\varepsilon} \diff x \int_{N_\varepsilon}^{y/m} f(x) \diff x \geq \int_{N_\varepsilon}^{y/m}x^{-1 +\varepsilon} \diff x \int_{N_\varepsilon}^y f(x) \diff x.
	\end{equation}
	From this follows that
	\begin{equation}
		\frac{\int_{N_\varepsilon}^{y/m} f(x) \diff x }{\int_{N_\varepsilon}^{y} f(x) \diff x } \geq \frac{(y/m)^\varepsilon - N_\varepsilon^\varepsilon}{y^\varepsilon - N_\varepsilon^\varepsilon}.
	\end{equation}
	Taking the $\limsup$ on both sides, and adding the finite contributions $\int_1^{N_\varepsilon} f(x) \diff x$ to the diverging integrals $\int_{N_\varepsilon}^{y/m} f(x) \diff x$ and $\int_{N_\varepsilon}^y f(x) \diff x$, we obtain
	\begin{equation}\label{eq:int-low-bound}
		\limsup_{y \to \infty} \frac{\int_{1}^{y/m} f(x) \diff x }{\int_{1}^{y} f(x) \diff x } \geq \frac{1}{m^\varepsilon}.
	\end{equation}
	Combining Eq.~\eqref{eq:star-dist-up-bound} and Eq.~\eqref{eq:int-low-bound} we get
	\begin{equation}
		\limsup_{n \to \infty}T_\star(\chi_n \to \chi_n \otimes \Phi_m) \leq 1 - \limsup_{y \to \infty} \frac{\int_{1}^{y/m} f(x) \diff x }{\int_{1}^{y} f(x) \diff x } \leq 1 - \frac{1}{m^\varepsilon}.
	\end{equation}
	This is true for all $\varepsilon > 0$, thus taking the limit $\varepsilon \to 0$ we obtain the desired result
	\begin{equation}
		\limsup_{n \to \infty}T_\star(\chi_n \to \chi_n \otimes \Phi_m) \leq \lim_{\varepsilon \to 0} 1 - \frac{1}{m^\varepsilon} = 0.
	\end{equation}
	This concludes the proof of the sufficient condition.
\end{proof}

The above result on the uniqueness of the van Dam and Hayden family can be expressed as follows: Any regular family of states $\Set{\chi_n}_{n \in \mbN}$ satisfying the conditions of Theorem~\ref{th:uniq-vDH} is a universal embezzling family if and only if $f$, the function associated to $\Set{\chi_n}_{n \in \mbN}$, is asymptotically close to $x^{-1}$, the function associated to the van Dam and Hayden family, where asymptotically close means $f \in \omega(x^{-1 -\varepsilon}) \cap o(x^{-1 + \varepsilon})$ for all $\varepsilon > 0$.

The assumption that $f(x)/x^\alpha$ is asymptotically monotonic for every $\alpha \in \mbR$ is crucial for our proof and does not follow from the monotonicity of $f$. There are functions
that are non-increasing, but oscillate asymptotically when multiplied by powers of $x$. An example is the function
\begin{equation}\label{eq:ex-non-incr-f}
	f(x)=\frac{1+(1+\sin\ln\ln x)\ln x}{x},
\end{equation}
which is non-increasing, but  $x f(x)$ oscillates between $1$ and $\infty$. Theorem~\ref{th:uniq-vDH} does not provide any information about  families of states associated to such functions, and it cannot be used to determine whether such families are universally embezzling or not.

Next, we prove two related propositions concerning regular families of states associated to non-decreasing functions.

\begin{proposition-app}\label{prop:f-incr-vdh-alpha}
	Let $f$ be a positive non-decreasing function such that $f(x)/x^\alpha$ is asymptotically non-increasing for at least one $\alpha >0$. Then the regular family of states $\Set{\chi_n}_{n \in \mbN}$ associated to $f$ (see Definition~\ref{def:reg-fam}) is not a universal embezzling family. Furthermore, there are no sequences $\Set{n_j}_{j \in \mbN}$ such that $\Set{\chi_{n_j}}_{j \in \mbN}$ is a universal embezzling family.
\end{proposition-app}
\begin{proof}
	Since $f(x)$ is non-decreasing and $\sum_{x=1}^{\infty} f(x)$ diverges, we can use Proposition~\ref{prop:sum-to-int} and Lemma~\ref{le:int-ineq}, to obtain
	\begin{equation}
		\begin{aligned}
			\liminf_{n \to \infty}T_\star(\chi_n \to \chi_n \otimes \Phi_m) &\ge \liminf_{y \to \infty}\frac{\max_{1 \leq a \leq y/m}\Set{\int_{a}^{am} g(x, y)\diff x}}{\int_{1}^{y} g(x, y) \diff x}\\
			&= \liminf_{y \to \infty}\frac{\max_{1 \leq a \leq y/m}\Set{\int^{y - a+1}_{y - am +1} f(x)\diff x}}{\int_{1}^{y}f(x) \diff x}\\
			&\geq \liminf_{y \to \infty}\frac{\Set{\int^{y(1 - 1/m)+1}_{1} f(x)\diff x}}{\int_{1}^{y}f(x) \diff x}\\
			&\geq \liminf_{y \to \infty}\frac{\Set{\int^{y(1 - 1/m)+1}_{1} x^\alpha\diff x}}{\int_{1}^{y}x^\alpha \diff x}\\
			&=\left(1 - \frac{1}{m}\right)^{\alpha + 1}.
		\end{aligned}
	\end{equation}
	The last inequality is based on Lemma~\ref{le:int-ineq} and is derived as in the previous cases. Using the remark after Definition~\ref{def:reg-fam}, we obtain that $\Set{\chi_n}_{n \in \mbN}$ is not a universal embezzling family. Furthermore, there are no sequences $\Set{n_j}_{j \in \mbN}$ such that $\Set{\chi_{n_j}}_{j \in \mbN}$ is a universal embezzling family.
\end{proof}
\begin{proposition-app}\label{prop:f-incr-vdh-k}
	Let $f$ be a positive non-decreasing function such that $f(x)/e^{kx}$ is asymptotically non-decreasing for at least one $k >0$. Then the regular family of states $\Set{\chi_n}_{n \in \mbN}$ associated to $f$ is not a universal embezzling family. Furthermore, there are no sequences $\Set{n_j}_{j \in \mbN}$ such that $\Set{\chi_{n_j}}_{j \in \mbN}$ is a universal embezzling family.
\end{proposition-app}
\begin{proof}
	From Eq.~\eqref{eq:closed-star-conv-dist} follows that
	\begin{equation}
		\begin{aligned}
			\liminf_{n \to \infty}T_\star(\chi_n \to \chi_n \otimes \Phi_m) &\geq \left(1 - \frac{1}{m}\right) \liminf_{n \to \infty} p_{1}^{(n)}\\
			&= \left(1 - \frac{1}{m}\right)\liminf_{n \to \infty} \frac{f(n)}{\sum_{x=1}^{n}f(x)}\\
			&= \left(1 - \frac{1}{m}\right)\liminf_{y \to \infty} \frac{f(y)}{\int_{1}^{y}f(x) \diff x}\\
			&\geq \left(1 - \frac{1}{m}\right)\liminf_{y \to \infty} \frac{\int_{y -1}^y f(x) \diff x}{\int_{1}^{y}f(x) \diff x}\\
			&\geq \left(1 - \frac{1}{m}\right)\liminf_{y \to \infty} \frac{\int_{y -1}^y e^{kx} \diff x}{\int_{1}^{y}e^{kx} \diff x}\\
			&= \left(1 - \frac{1}{m}\right)(1 - e^{-k}).
		\end{aligned}
	\end{equation}
	Also here the last inequality is due to Lemma~\ref{le:int-ineq} and the family of states is not a universal embezzling family. Furthermore, there are no sequences $\Set{n_j}_{j \in \mbN}$ such that $\Set{\chi_{n_j}}_{j \in \mbN}$ is a universal embezzling family.
\end{proof}

\section{Asymptotically Regular Families}

In the definition of universal embezzling families, Definition~\ref{def:uef-app}, and in all the results about embezzlement, only the asymptotic behaviour of a family of states $\Set{\chi_n}_{n \in \mbN}$ is relevant. This motivates the following definition, which is a generalization of regular families.
\begin{definition-app}\label{def:a-reg-fam}
	A family of states $\Set{\chi_n}_{n \in \mbN}$ is called asymptotically regular  if there exists an asymptotically monotonic function $f\colon \mbN \to (0, \infty)$ such that 
	\begin{equation}
		\ket{\chi_n} = \frac{1}{\sqrt{F_n}}\sum_{x = 1}^{n}\sqrt{f(x)} \Ket{xx} \quad \text{for all} \quad n \in \mbN,
	\end{equation}
	where $F_n = \sum_{x=1}^{n} f(x)$.
\end{definition-app}
We next show that our results hold for asymptotically regular families too. To this end, we start with the following theorem.

\begin{theorem-app}\label{th:a-reg-fam}
	Let $\Set{\chi_n}_{n \in \mbN}$ be an asymptotically regular family and let $f$ be the function associated to it (see Definition~\ref{def:a-reg-fam}). Then one can construct a function $\tilde{f}$ that satisfies 
	\begin{enumerate}
		\item $\tilde f$ corresponds to a regular family $\Set{\tilde\chi_n}_{n \in \mbN}$ (see Definition~\ref{def:reg-fam}),
		\item $\lim_{x \to \infty} \tilde{f}(x)/f(x) =1$,
		\item $\Set{\tilde\chi_n}_{n \in \mbN}$ is a universal embezzling family if and only if $\Set{\chi_n}_{n \in \mbN}$ is a universal embezzling family,
		\item $\Set{\tilde\chi_n}_{n \in \mbN}$ contains a universal embezzling subfamily if and only if $\Set{\chi_n}_{n \in \mbN}$ contains a universal embezzling subfamily.
	\end{enumerate}
\end{theorem-app}
\begin{proof}
	Since $f$ is by assumption asymptotically monotonic,  $\lim_{x \to\infty} f(x)$ exists in the extended domain $[0, \infty]$. We first study the case $\lim_{x \to\infty} f(x) = l$, with $0 < l \in \mbR$. In this case, we choose $\tilde{f}(x) = l$, which corresponds to a regular family. Moreover, from the definition of $\tilde{f}$ follows that $\lim_{x \to \infty}f(x)/ \tilde{f}(x) =1$. We also notice that for every $0 < \varepsilon< l$,  there exists an $N$ such that $l -\varepsilon < f(x) < l + \varepsilon$ for all $x \geq N$. This implies that $\sum_{x = 1}^\infty f(x) \geq \sum_{x = N}^\infty f(x) \geq \sum_{x = N}^\infty (l - \varepsilon) = \infty$. Similarly, $\sum_{x = 1}^\infty \tilde{f}(x) = \sum_{x = 1}^\infty l = \infty$. Since $\lim_{n \to \infty} F_n^{-1} = \lim_{n \to \infty} \tilde{F}_n^{-1} = 0$, we obtain that $\lim_{n\to\infty}p_1^{(n)}=0$ (where $\p^{(n)}$ are the Schmidt coefficients of $\chi_n$), and therefore, according to Eq.~\eqref{eq:closed-star-conv-dist},
	\begin{equation}\label{eq:liming-f-l}
		\begin{aligned}
			\liminf_{n \to \infty} T_\star(\chi_n \to \chi_n \otimes \Phi_m) &= \liminf_{n \to \infty} \max_{k \in [n]}\left\{\N{\p^{(n)}}_{(k)} - \N{\p^{(n)}}_{(a_k)} - \frac{b_k}{m} p^{(n)}_{ a_k+1}\right\} \\
			&= \liminf_{n \to \infty} \max_{k \in \Set{N, \dots, n}}\left\{\sum_{x = a_k +1}^k p^{(n)}_x - \frac{b_k}{m} p^{(n)}_{ a_k+1}\right\} \\
			&= \liminf_{n \to \infty} \max_{k \in \Set{N, \dots, n}}\left\{\sum_{x = a_k +1}^k p^{(n)}_x \right\},
		\end{aligned}
	\end{equation}
	where $a_k = \floor{k/m}$ and $b_k = k - m a_k$. The Schmidt coefficients are by definition non-increasing, i.e., obtained by reordering $\Set{f(x)/F_n}_{x \in \mbN}$. Since there are at most $N$ natural numbers $x$ that do not satisfy the condition $l - \varepsilon<f(x)< l + \varepsilon$, there are at most $N$ Schmidt coefficients that do not satisfy $\frac{l - \varepsilon}{F_n} < p^{(n)}_x < \frac{l + \varepsilon}{F_n}$. We call the set of indices corresponding to these Schmidt coefficients $A$ (thus $\A{A} \leq N$) and observe that for any $a, b \in \Set{N, \dots, n}$ such that $a \leq b$,
	\begin{equation}
		\begin{aligned}
			\sum_{x = a}^b p^{(n)}_x &= \sum_{x \in \Set{a, \dots, b} \setminus A} p^{(n)}_x + \sum_{x \in \Set{a, \dots, b} \cap A} p^{(n)}_x\\
			&= \sum_{x \in \Set{a, \dots, b} \setminus A} p^{(n)}_x + \sum_{x \in \Set{a, \dots, b} \cap A} \frac{l}{F_n} + \sum_{x \in \Set{a, \dots, b} \cap A} \left(p^{(n)}_x -\frac{l}{F_n}\right) .
		\end{aligned}
	\end{equation}
	This implies that
	\begin{equation}
		\begin{aligned}
			\sum_{x = a}^b \frac{(l - \varepsilon)}{F_n} + \sum_{x \in \Set{a, \dots, b} \cap A} \left(p^{(n)}_x -\frac{l}{F_n}\right) &< \sum_{x = a}^b p^{(n)}_x \\
			&< \sum_{x = a}^b \frac{(l + \varepsilon)}{F_n} + \sum_{x \in \Set{a, \dots, b} \cap A} \left(p^{(n)}_x -\frac{l}{F_n}\right),
		\end{aligned}
	\end{equation}
	and therefore
	\begin{equation}
		\begin{aligned}
			\liminf_{n \to \infty} \max_{k \in \Set{N, \dots, n}}\left\{\sum_{x = a_k +1}^k \frac{l-\varepsilon}{F_n} \right\} &\leq \liminf_{n \to \infty} \max_{k \in \Set{N, \dots, n}}\left\{\sum_{x = a_k +1}^k p^{(n)}_x \right\} \\
			&\leq \liminf_{n \to \infty} \max_{k \in \Set{N, \dots, n}}\left\{\sum_{x = a_k +1}^k \frac{l+\varepsilon}{F_n}\right\} ,
		\end{aligned}
	\end{equation}
	which implies
	\begin{equation}
		\begin{aligned}
			(l - \varepsilon)\liminf_{n \to \infty} \frac{n - \floor{n/m}}{F_n}  &\leq \liminf_{n \to \infty} \max_{k \in \Set{N, \dots, n}}\left\{\sum_{x = a_k +1}^k p^{(n)}_x \right\} \\
			&\leq   (l + \varepsilon)\liminf_{n \to \infty} \frac{n - \floor{n/m}}{F_n}.
		\end{aligned}
	\end{equation}
	Observing that $\sum_{x = 1}^{N} f(x) + (l - \varepsilon)(n - N) < F_n < \sum_{x = 1}^{N} f(x) + (l + \varepsilon)(n - N)$, we obtain
	\begin{equation}\label{eq:low-up-bounds-f-l}
		\frac{l - \varepsilon}{l + \varepsilon}\left(1 - \frac{1}{m}\right) \leq \liminf_{n \to \infty} \max_{k \in \Set{N, \dots, n}}\left\{\sum_{x = a_k +1}^k p^{(n)}_x \right\} \leq   \frac{l + \varepsilon}{l - \varepsilon}\left(1 - \frac{1}{m}\right) .
	\end{equation}
	Since Eq.~\eqref{eq:low-up-bounds-f-l} holds for every $\varepsilon > 0$, we conclude that
	\begin{equation}
		\liminf_{n \to \infty} \max_{k \in \Set{N, \dots, n}}\left\{\sum_{x = a_k +1}^k p^{(n)}_x \right\} = 1 - \frac{1}{m}.
	\end{equation}
	Inserting this result into Eq.~\eqref{eq:liming-f-l}, we have
	\begin{equation}
		\liminf_{x \to \infty}T_\star(\chi_n \to \chi_n \otimes \Phi_m) = 1 - \frac{1}{m}.
	\end{equation}
	The function $\tilde{f}$ is a rescaling of $f(x) = x^0$, which we already studied in Section~\ref{sec:gen-van-dam}. This implies that the family $\Set{\tilde{\chi}_n}_{n \in \mbN}$ is equal to the family $\Set{\chi^{(0)}_n}_{n \in \mbN}$ and
	\begin{equation}
		\begin{aligned}
			\liminf_{x \to \infty}T_\star(\tilde\chi_n \to \tilde\chi_n \otimes \Phi_m) &= \liminf_{x \to \infty}T_\star(\chi_n^{(0)} \to \chi_n^{(0)} \otimes \Phi_m) \\
			&= 1- \frac{1}{m} \\
			&= \liminf_{x \to \infty}T_\star(\chi_n \to \chi_n \otimes \Phi_m) .
		\end{aligned}
	\end{equation}
	This proves that neither $f$ nor $\tilde f$ corresponds to families with universally embezzling subfamilies (and are therefore also not universally embezzling themselves).
	
	We consider now the case $\lim_{x \to \infty} f(x) = 0$, thus $f$ is asymptotically non-increasing. Let $N$ be such that $f$ is non-increasing on $(N, \infty)$ and let $a = \min_{x \in \Set{1, \dots, N}} f(x)$. Since $f$ is positive, $a > 0$. Let $M> N$ be such that $f(x) < a$ for every $x \geq M$ (such $M$ exists because $\lim_{x \to \infty} f(x) = 0$). In this case, we define  $\tilde{f}$ as
	\begin{equation}
		\tilde{f}(x) = \begin{cases}
			f(M) & \text{if } x \leq M,\\
			f(x) & \text{if } x > M.
		\end{cases}
	\end{equation}
	Clearly, the family $\Set{\tilde\chi_n}_{n \in \mbN}$ associated to it is regular and $\lim_{x \to \infty} f(x)/\tilde{f}(x) = 1$. Furthermore, the ordered Schmidt coefficients satisfy for all $x \geq M$
	\begin{equation}\label{eq:sc-f-tf-0}
		p^{(n)}_{x} = \frac{f(x)}{F_n} = \frac{\tilde{f}(x)}{F_n} = \tilde{p}^{(n)}_{x} \frac{\tilde{F}_n}{F_n} .
	\end{equation}
	If $F_n$ converges, $\tilde{F}_n$ converges too and by Corollary~\ref{co:emb-nec-cond} neither $\Set{\chi_n}_{n \in \mbN}$ nor $\Set{\tilde\chi_n}_{n \in \mbN}$ are universal embezzling families (and do not contain universal embezzling subfamilies).  If instead $F_n$ diverges, so does $\tilde{F}_n$ and
	\begin{equation}
		\lim_{n \to \infty} \frac{\tilde{F}_n}{F_n} = \lim_{n \to \infty} \frac{\sum_{x = 1}^n \tilde{f}(x)}{\sum_{x = 1}^n f(x)} = \lim_{n \to \infty} \frac{\sum_{x = M}^n \tilde{f}(x)}{\sum_{x = M}^n f(x)} = 1.
	\end{equation}
	Using Eq.~\eqref{eq:closed-star-conv-dist} and Eq.~\eqref{eq:sc-f-tf-0}, we obtain
	\begin{equation}
		\begin{aligned}
			\liminf_{n \to \infty}\, &T_\star(\chi_n \to \chi_n \otimes \Phi_m)\\
			 &= \liminf_{n \to \infty} \max_{k \in [n]}\left\{\N{\p^{(n)}}_{(k)} - \N{\p^{(n)}}_{(a_k)} - \frac{b_k}{m} p^{(n)}_{ a_k+1}\right\}\\
			&= \liminf_{n \to \infty} \max_{k \in \Set{mM, \dots, n}}\left\{\N{\p^{(n)}}_{(k)} - \N{\p^{(n)}}_{(a_k)} - \frac{b_k}{m} p^{(n)}_{ a_k+1}\right\}\\
			&= \liminf_{n \to \infty}\frac{\tilde{F}_n}{F_n} \max_{k \in \Set{mM, \dots, n}}\left\{\N{\tilde{\p}^{(n)}}_{(k)} - \N{\tilde{\p}^{(n)}}_{(a_k)} - \frac{b_k}{m} \tilde{p}^{(n)}_{ a_k+1}\right\}\\
			&= \liminf_{n \to \infty} \max_{k \in [n]}\left\{\N{\tilde{\p}^{(n)}}_{(k)} - \N{\tilde{\p}^{(n)}}_{(a_k)} - \frac{b_k}{m} \tilde{p}^{(n)}_{ a_k+1}\right\}\\
			&=\liminf_{n \to \infty} T_\star(\tilde\chi_n \to \tilde\chi_n \otimes \Phi_m) ,
		\end{aligned}
	\end{equation}
	where again $a_k = \floor{k/m}$ and $b_k = k - m a_k$. Also in this case, we therefore proved that $\Set{\chi_n}_{n \in \mbN}$ is a universal embezzling family if and only if $\Set{\tilde{\chi}_n}_{n \in \mbN}$ is a universal embezzling family (and the same holds for subfamilies).
	
	Lastly, we consider the case when $\lim_{x \to \infty} f(x) = \infty$, and $f$ is asymptotically non-decreasing. Analogously to the previous case, let $N$ be such that $f$ is non-decreasing for $x \in (N, \infty)$. Let $a$ be the maximum of $f(x)$ for $x \in \Set{1, \dots, N}$, and let $M > N$ be such that $f(x) > a$ for $x \geq M$. Also here, we define $\tilde{f}$ via
	\begin{equation}
		\tilde{f} = \begin{cases}
			f(M) & \text{if } x \leq M,\\
			f(x) & \text{if } x > M.
		\end{cases}
	\end{equation}
	The family of states $\Set{\tilde\chi_n}_{n \in \mbN}$ is regular, $\lim_{x \to \infty} \tilde{f}(x)/f(x) =1$, and $\lim_{x \to \infty}\tilde{F}_n/ F_n = 1$. The Schmidt coefficients associated to $\chi_n$ and $\tilde\chi_n$ are related by
	\begin{equation}\label{eq:sc-f-tf-inf}
		p^{(n)}_x = \frac{f(n + 1 -x)}{F_n} = \frac{\tilde{f}(n +1 -x)}{F_n} = \tilde{p}^{(n)}_{x} \frac{\tilde{F}_n}{F_n}  \quad \forall x \leq n - M .
	\end{equation}
	Using Eq.~\eqref{eq:closed-star-conv-dist} again and the relation between Schmidt coefficients derived in Eq.~\eqref{eq:sc-f-tf-inf} we obtain
	\begin{equation}
		\begin{aligned}
			\liminf_{n \to \infty}\, &T_\star(\chi_n \to \chi_n \otimes \Phi_m) \\
			&= \liminf_{n \to \infty} \max_{k \in [n]}\left\{\N{\p^{(n)}}_{(k)} - \N{\p^{(n)}}_{(a_k)} - \frac{b_k}{m} p^{(n)}_{ a_k+1}\right\}\\
			&= \liminf_{n \to \infty} \max_{k \in [n - M]}\left\{\N{\p^{(n)}}_{(k)} - \N{\p^{(n)}}_{(a_k)} - \frac{b_k}{m} p^{(n)}_{ a_k+1}\right\}\\
			&= \liminf_{n \to \infty}\frac{\tilde{F}_n}{F_n} \max_{k \in [n - M]}\left\{\N{\tilde{\p}^{(n)}}_{(k)} - \N{\tilde{\p}^{(n)}}_{(a_k)} - \frac{b_k}{m} \tilde{p}^{(n)}_{ a_k+1}\right\}\\
			&= \liminf_{n \to \infty} \max_{k \in [n]}\left\{\N{\tilde{\p}^{(n)}}_{(k)} - \N{\tilde{\p}^{(n)}}_{(a_k)} - \frac{b_k}{m} \tilde{p}^{(n)}_{ a_k+1}\right\}\\
			&=\liminf_{n \to \infty} T_\star(\tilde\chi_n \to \tilde\chi_n \otimes \Phi_m).
		\end{aligned}
	\end{equation}
	This proves the theorem.
\end{proof}
Thanks to Theorem~\ref{th:a-reg-fam}, Theorem~\ref{th:uniq-vDH} also holds for asymptotically regular families.
\begin{corollary-app}\label{co:a-reg-fam}
	Let $f$ be a positive asymptotically non-increasing function such that  $f(x)/x^\alpha$ is asymptotically monotonic for all $\alpha \in \mbR$ and let $\Set{\chi_n}_{n \in \mbN}$ be the asymptotically regular family of states associated to $f$ (see Definition~\ref{def:a-reg-fam}). Then $\Set{\chi_n}_{n \in \mbN}$ is a universal embezzling family if and only if $f \in \omega(x^{-1-\varepsilon})\cap o(x^{-1 +\varepsilon})$ for all $\varepsilon>0$ and $\sum_{x = 1}^\infty f(x) = \infty$. Furthermore, if $f \notin \omega(x^{-1-\varepsilon})\cap o(x^{-1 +\varepsilon})$ for at least one $\varepsilon>0$, then $\Set{\chi_{n_j}}_{j \in \mbN}$, where $\Set{n_j}_{j \in \mbN}$ is any sequence of natural numbers, is not a universal embezzling family.
\end{corollary-app}
\begin{proof}
	Combine Theorem~\ref{th:a-reg-fam} and Theorem~\ref{th:uniq-vDH}.
\end{proof}
For the same reasons, Proposition~\ref{prop:f-incr-vdh-alpha} and Proposition~\ref{prop:f-incr-vdh-k} also hold for asymptotically regular families.
\begin{corollary-app}
	Let $f$ be a positive asymptotically non-decreasing function such that $f(x)/x^\alpha$ is asymptotically non-increasing for at least one $\alpha >0$. Then the asymptotically regular family of states $\Set{\chi_n}_{n \in \mbN}$ associated to $f$ is not a universal embezzling family. Furthermore, there are no sequences $\Set{n_j}_{j \in \mbN}$ such that $\Set{\chi_{n_j}}_{j \in \mbN}$ is a universal embezzling family.
\end{corollary-app}
\begin{corollary-app}
	Let $f$ be a positive asymptotically non-decreasing function such that $f(x)/e^{kx}$ is asymptotically non-decreasing for at least one $k >0$. Then the asymptotically regular family of states $\Set{\chi_n}_{n \in \mbN}$ associated with $f$ is not a universal embezzling family. Furthermore, there are no sequences $\Set{n_j}_{j \in \mbN}$ such that $\Set{\chi_{n_j}}_{j \in \mbN}$ is a universal embezzling family.
\end{corollary-app}


\begin{thebibliography}{10}
	
	\bibitem{BBPS96}
	Charles~H. Bennett, Herbert~J. Bernstein, Sandu Popescu, and Benjamin
	Schumacher.
	\newblock ``Concentrating partial entanglement by local operations''.
	\newblock \href{https://dx.doi.org/10.1103/PhysRevA.53.2046}{Phys. Rev. A {\bf
			53}, 2046--2052}~(1996).
	
	\bibitem{VPRK97}
	V.~Vedral, M.~B. Plenio, M.~A. Rippin, and P.~L. Knight.
	\newblock ``Quantifying entanglement''.
	\newblock \href{https://dx.doi.org/10.1103/PhysRevLett.78.2275}{Phys. Rev.
		Lett. {\bf 78}, 2275--2279}~(1997).
	
	\bibitem{Nie99}
	M.~A. Nielsen.
	\newblock ``Conditions for a class of entanglement transformations''.
	\newblock \href{https://dx.doi.org/10.1103/PhysRevLett.83.436}{Phys. Rev. Lett.
		{\bf 83}, 436--439}~(1999).
	
	\bibitem{BPRST00}
	Charles~H. Bennett, Sandu Popescu, Daniel Rohrlich, John~A. Smolin, and
	Ashish~V. Thapliyal.
	\newblock ``Exact and asymptotic measures of multipartite pure-state
	entanglement''.
	\newblock \href{https://dx.doi.org/10.1103/PhysRevA.63.012307}{Phys. Rev. A
		{\bf 63}, 012307}~(2000).
	
	\bibitem{PV07}
	Martin~B. Plenio and Shashank Virmani.
	\newblock ``An introduction to entanglement measures''.
	\newblock \href{https://dx.doi.org/10.26421/QIC7.1-2-1}{Quantum Inf. Comput.
		{\bf 7}, 1--51}~(2007).
	
	\bibitem{HHHH09}
	Ryszard Horodecki, Pawe\l{} Horodecki, Micha\l{} Horodecki, and Karol
	Horodecki.
	\newblock ``Quantum entanglement''.
	\newblock \href{https://dx.doi.org/10.1103/RevModPhys.81.865}{Rev. Mod. Phys.
		{\bf 81}, 865--942}~(2009).
	
	\bibitem{EPR35}
	A.~Einstein, B.~Podolsky, and N.~Rosen.
	\newblock ``Can quantum-mechanical description of physical reality be
	considered complete?''.
	\newblock \href{https://dx.doi.org/10.1103/PhysRev.47.777}{Phys. Rev. {\bf 47},
		777--780}~(1935).
	
	\bibitem{BB84a}
	C.~H. Bennet and G~Brassard.
	\newblock ``Quantum cryptography: Public key distribution and coin tossing''.
	\newblock In Proceedings of the IEEE International Conference on Computers,
	Systems and Signal Processing.
	\newblock Volume~1, pages 175--179.
	\newblock IEEE Bangalore~(1984).
	\newblock
	url:~\url{https://web.archive.org/web/20200130165639/http://researcher.watson.ibm.com/researcher/files/us-bennetc/BB84highest.pdf}.
	
	\bibitem{BB84b}
	Charles~H. Bennett and Gilles Brassard.
	\newblock ``Quantum cryptography: Public key distribution and coin tossing''.
	\newblock
	\href{https://dx.doi.org/https://doi.org/10.1016/j.tcs.2014.05.025}{Theor.
		Comput. Sci. {\bf 560}, 7--11}~(2014).
	
	\bibitem{BWS92}
	Charles~H. Bennett and Stephen~J. Wiesner.
	\newblock ``Communication via one- and two-particle operators on
	{E}instein-{P}odolsky-{R}osen states''.
	\newblock \href{https://dx.doi.org/10.1103/PhysRevLett.69.2881}{Phys. Rev.
		Lett. {\bf 69}, 2881--2884}~(1992).
	
	\bibitem{BBCJPW93}
	Charles~H. Bennett, Gilles Brassard, Claude Cr{\'e}peau, Richard Jozsa, Asher
	Peres, and William~K. Wootters.
	\newblock ``Teleporting an unknown quantum state via dual classical and
	{E}instein-{P}odolsky-{R}osen channels''.
	\newblock \href{https://dx.doi.org/10.1103/PhysRevLett.70.1895}{Phys. Rev. Lett
		{\bf 70}, 1895--1899}~(1993).
	
	\bibitem{BdVSW96}
	Charles~H. Bennett, David~P. DiVincenzo, John~A. Smolin, and William~K.
	Wootters.
	\newblock ``Mixed-state entanglement and quantum error correction''.
	\newblock \href{https://dx.doi.org/10.1103/PhysRevA.54.3824}{Phys. Rev. A {\bf
			54}, 3824--3851}~(1996).
	
	\bibitem{BSST99}
	Charles~H. Bennett, Peter~W. Shor, John~A. Smolin, and Ashish~V. Thapliyal.
	\newblock ``Entanglement-assisted classical capacity of noisy quantum
	channels''.
	\newblock \href{https://dx.doi.org/10.1103/PhysRevLett.83.3081}{Phys. Rev.
		Lett. {\bf 83}, 3081--3084}~(1999).
	
	\bibitem{BSST06}
	C.~H. Bennett, P.~W. Shor, J.~A. Smolin, and A.~V. Thapliyal.
	\newblock ``Entanglement-assisted capacity of a quantum channel and the reverse
	{S}hannon theorem''.
	\newblock \href{https://dx.doi.org/10.1109/TIT.2002.802612}{IEEE Trans. Inf.
		Theory {\bf 48}, 2637--2655}~(2006).
	
	\bibitem{BDHSW14}
	Charles~H. Bennett, Igor Devetak, Aram~W. Harrow, Peter~W. Shor, and Andreas
	Winter.
	\newblock ``The quantum reverse {S}hannon theorem and resource tradeoffs for
	simulating quantum channels''.
	\newblock \href{https://dx.doi.org/10.1109/TIT.2014.2309968}{IEEE Trans. Inf.
		Theory {\bf 60}, 2926--2959}~(2014).
	
	\bibitem{Pop94}
	Sandu Popescu.
	\newblock ``Bell's inequalities versus teleportation: What is nonlocality?''.
	\newblock \href{https://dx.doi.org/10.1103/PhysRevLett.72.797}{Phys. Rev. Lett.
		{\bf 72}, 797--799}~(1994).
	
	\bibitem{HHH99}
	Micha\l{} Horodecki, Pawe\l{} Horodecki, and Ryszard Horodecki.
	\newblock ``General teleportation channel, singlet fraction, and
	quasidistillation''.
	\newblock \href{https://dx.doi.org/10.1103/PhysRevA.60.1888}{Phys. Rev. A {\bf
			60}, 1888--1898}~(1999).
	
	\bibitem{BCR11}
	Mario Berta, Matthias Christandl, and Renato Renner.
	\newblock ``The quantum reverse shannon theorem based on one-shot information
	theory''.
	\newblock \href{https://dx.doi.org/10.1007/s00220-011-1309-7}{Commun. Math.
		Phys. {\bf 306}, 579--615}~(2011).
	
	\bibitem{CG19}
	Eric Chitambar and Gilad Gour.
	\newblock ``Quantum resource theories''.
	\newblock \href{https://dx.doi.org/10.1103/RevModPhys.91.025001}{Rev. Mod.
		Phys. {\bf 91}, 025001}~(2019).
	
	\bibitem{CFS16}
	Bob Coecke, Tobias Fritz, and Robert~W. Spekkens.
	\newblock ``A mathematical theory of resources''.
	\newblock
	\href{https://dx.doi.org/https://doi.org/10.1016/j.ic.2016.02.008}{Inf.
		Comput. {\bf 250}, 59--86}~(2016).
	
	\bibitem{JP99A}
	Daniel Jonathan and Martin~B. Plenio.
	\newblock ``Minimal conditions for local pure-state entanglement
	manipulation''.
	\newblock \href{https://dx.doi.org/10.1103/PhysRevLett.83.1455}{Phys. Rev.
		Lett. {\bf 83}, 1455--1458}~(1999).
	
	\bibitem{VJN00}
	Guifr\'e Vidal, Daniel Jonathan, and M.~A. Nielsen.
	\newblock ``Approximate transformations and robust manipulation of bipartite
	pure-state entanglement''.
	\newblock \href{https://dx.doi.org/10.1103/PhysRevA.62.012304}{Phys. Rev. A
		{\bf 62}, 012304}~(2000).
	
	\bibitem{Uhl76}
	A.~Uhlmann.
	\newblock ``The ``transition probability'' in the state space of a
	{*}-algebra''.
	\newblock
	\href{https://dx.doi.org/https://doi.org/10.1016/0034-4877(76)90060-4}{Rep.
		Math. Phys. {\bf 9}, 273--279}~(1976).
	
	\bibitem{Joz94}
	Richard Jozsa.
	\newblock ``Fidelity for mixed quantum states''.
	\newblock \href{https://dx.doi.org/10.1080/09500349414552171}{J. Mod. Opt {\bf
			41}, 2315--2323}~(1994).
	
	\bibitem{FVdG99}
	C.A. Fuchs and J.~van~de Graaf.
	\newblock ``Cryptographic distinguishability measures for quantum-mechanical
	states''.
	\newblock \href{https://dx.doi.org/10.1109/18.761271}{IEEE Trans. Inf. Theory
		{\bf 45}, 1216--1227}~(1999).
	
	\bibitem{TCR10}
	Marco Tomamichel, Roger Colbeck, and Renato Renner.
	\newblock ``Duality between smooth min- and max-entropies''.
	\newblock \href{https://dx.doi.org/10.1109/TIT.2010.2054130}{IEEE Trans. Inf.
		Theory {\bf 56}, 4674--4681}~(2010).
	
	\bibitem{GT20}
	Gilad Gour and Marco Tomamichel.
	\newblock ``Optimal extensions of resource measures and their applications''.
	\newblock \href{https://dx.doi.org/10.1103/PhysRevA.102.062401}{Phys. Rev. A
		{\bf 102}, 062401}~(2020).
	
	\bibitem{WW19}
	Xin Wang and Mark~M. Wilde.
	\newblock ``Resource theory of asymmetric distinguishability''.
	\newblock \href{https://dx.doi.org/10.1103/PhysRevResearch.1.033170}{Phys. Rev.
		Res. {\bf 1}, 033170}~(2019).
	
	\bibitem{SCG20}
	Gaurav Saxena, Eric Chitambar, and Gilad Gour.
	\newblock ``Dynamical resource theory of quantum coherence''.
	\newblock \href{https://dx.doi.org/10.1103/PhysRevResearch.2.023298}{Phys. Rev.
		Res. {\bf 2}, 023298}~(2020).
	
	\bibitem{SG22}
	Gaurav Saxena and Gilad Gour.
	\newblock ``Quantifying multiqubit magic channels with completely
	stabilizer-preserving operations''.
	\newblock \href{https://dx.doi.org/10.1103/PhysRevA.106.042422}{Phys. Rev. A
		{\bf 106}, 042422}~(2022).
	
	\bibitem{Gou22}
	Gilad Gour.
	\newblock ``Role of quantum coherence in thermodynamics''.
	\newblock \href{https://dx.doi.org/10.1103/PRXQuantum.3.040323}{PRX Quantum
		{\bf 3}, 040323}~(2022).
	
	\bibitem{vDH03}
	Wim van Dam and Patrick Hayden.
	\newblock ``Universal entanglement transformations without communication''.
	\newblock \href{https://dx.doi.org/10.1103/PhysRevA.67.060302}{Phys. Rev. A
		{\bf 67}, 060302}~(2003).
	
	\bibitem{LW14}
	Debbie Leung and Bingjie Wang.
	\newblock ``Characteristics of universal embezzling families''.
	\newblock \href{https://dx.doi.org/10.1103/PhysRevA.90.042331}{Phys. Rev. A
		{\bf 90}, 042331}~(2014).
	
	\bibitem{CLP17}
	Richard Cleve, Li~Liu, and Vern~I. Paulsen.
	\newblock ``Perfect embezzlement of entanglement''.
	\newblock \href{https://dx.doi.org/10.1063/1.4974818}{J. Math. Phys. {\bf 58},
		012204}~(2017).
	
	\bibitem{JP99O}
	Daniel Jonathan and Martin~B. Plenio.
	\newblock ``Entanglement-assisted local manipulation of pure quantum states''.
	\newblock \href{https://dx.doi.org/10.1103/PhysRevLett.83.3566}{Phys. Rev.
		Lett. {\bf 83}, 3566--3569}~(1999).
	
	\bibitem{LRS23}
	Ludovico Lami, Bartosz Regula, and Alexander Streltsov.
	\newblock ``Catalysis cannot overcome bound entanglement''~(2023).
	\newblock  \href{http://arxiv.org/abs/2305.03489}{arXiv:2305.03489}.
	
	\bibitem{GKS23}
	Ray Ganardi, Tulja~Varun Kondra, and Alexander Streltsov.
	\newblock ``Catalytic and asymptotic equivalence for quantum
	entanglement''~(2023).
	\newblock  \href{http://arxiv.org/abs/2305.03488}{arXiv:2305.03488}.
	
	\bibitem{SN23}
	Jeongrak Son and Nelly H.~Y. Ng.
	\newblock ``A hierarchy of thermal processes collapses under
	catalysis''~(2023).
	\newblock  \href{http://arxiv.org/abs/2303.13020}{arXiv:2303.13020}.
	
	\bibitem{LN23}
	Seok~Hyung Lie and Nelly H.~Y. Ng.
	\newblock ``Catalysis always degrades external quantum correlations''~(2023).
	\newblock  \href{http://arxiv.org/abs/2303.02376}{arXiv:2303.02376}.
	
	\bibitem{LWW23}
	Lauritz van Luijk, Reinhard~F. Werner, and Henrik Wilming.
	\newblock ``Covariant catalysis requires correlations and good quantum
	reference frames degrade little''~(2023).
	\newblock  \href{http://arxiv.org/abs/2301.09877}{arXiv:2301.09877}.
	
	\bibitem{DGKS23}
	Chandan Datta, Ray Ganardi, Tulja~Varun Kondra, and Alexander Streltsov.
	\newblock ``Is there a finite complete set of monotones in any quantum resource
	theory?''.
	\newblock \href{https://dx.doi.org/10.1103/PhysRevLett.130.240204}{Phys. Rev.
		Lett. {\bf 130}, 240204}~(2023).
	
	\bibitem{SN22}
	Jeongrak Son and Nelly H.~Y. Ng.
	\newblock ``Catalysis in action via elementary thermal operations''~(2022).
	\newblock  \href{http://arxiv.org/abs/2209.15213}{arXiv:2209.15213}.
	
	\bibitem{LPB23}
	Patryk Lipka-Bartosik, Mart\'{\i} Perarnau-Llobet, and Nicolas Brunner.
	\newblock ``Operational definition of the temperature of a quantum state''.
	\newblock \href{https://dx.doi.org/10.1103/PhysRevLett.130.040401}{Phys. Rev.
		Lett. {\bf 130}, 040401}~(2023).
	
	\bibitem{KL22}
	Kamil Korzekwa and Matteo Lostaglio.
	\newblock ``Optimizing thermalization''.
	\newblock \href{https://dx.doi.org/10.1103/PhysRevLett.129.040602}{Phys. Rev.
		Lett. {\bf 129}, 040602}~(2022).
	
	\bibitem{DKMS22}
	Chandan Datta, Tulja~Varun Kondra, Marek Miller, and Alexander Streltsov.
	\newblock ``Entanglement catalysis for quantum states and noisy
	channels''~(2022).
	\newblock  \href{http://arxiv.org/abs/2202.05228}{arXiv:2202.05228}.
	
	\bibitem{LS21}
	Patryk Lipka-Bartosik and Paul Skrzypczyk.
	\newblock ``Catalytic quantum teleportation''.
	\newblock \href{https://dx.doi.org/10.1103/PhysRevLett.127.080502}{Phys. Rev.
		Lett. {\bf 127}, 080502}~(2021).
	
	\bibitem{KDS21}
	Tulja~Varun Kondra, Chandan Datta, and Alexander Streltsov.
	\newblock ``Catalytic transformations of pure entangled states''.
	\newblock \href{https://dx.doi.org/10.1103/PhysRevLett.127.150503}{Phys. Rev.
		Lett. {\bf 127}, 150503}~(2021).
	
	\bibitem{Kar21}
	Martti Karvonen.
	\newblock ``Neither contextuality nor nonlocality admits catalysts''.
	\newblock \href{https://dx.doi.org/10.1103/PhysRevLett.127.160402}{Phys. Rev.
		Lett. {\bf 127}, 160402}~(2021).
	
	\bibitem{Wil21}
	H.~Wilming.
	\newblock ``Entropy and reversible catalysis''.
	\newblock \href{https://dx.doi.org/10.1103/PhysRevLett.127.260402}{Phys. Rev.
		Lett. {\bf 127}, 260402}~(2021).
	
	\bibitem{LJ21}
	Seok~Hyung Lie and Hyunseok Jeong.
	\newblock ``Randomness for quantum channels: Genericity of catalysis and
	quantum advantage of uniformness''.
	\newblock \href{https://dx.doi.org/10.1103/PhysRevResearch.3.013218}{Phys. Rev.
		Res. {\bf 3}, 013218}~(2021).
	
	\bibitem{LS21a}
	Patryk Lipka-Bartosik and Paul Skrzypczyk.
	\newblock ``All states are universal catalysts in quantum thermodynamics''.
	\newblock \href{https://dx.doi.org/10.1103/PhysRevX.11.011061}{Phys. Rev. X
		{\bf 11}, 011061}~(2021).
	
	\bibitem{SS21}
	Naoto Shiraishi and Takahiro Sagawa.
	\newblock ``Quantum thermodynamics of correlated-catalytic state conversion at
	small scale''.
	\newblock \href{https://dx.doi.org/10.1103/PhysRevLett.126.150502}{Phys. Rev.
		Lett. {\bf 126}, 150502}~(2021).
	
	\bibitem{LW23}
	Patryk Lipka-Bartosik, Henrik Wilming, and Nelly H.~Y. Ng.
	\newblock ``Catalysis in quantum information theory''~(2023).
	\newblock  \href{http://arxiv.org/abs/2306.00798}{arXiv:2306.00798}.
	
	\bibitem{DKMS22a}
	Chandan Datta, Tulja~Varun Kondra, Marek Miller, and Alexander Streltsov.
	\newblock ``Catalysis of entanglement and other quantum resources''~(2022).
	\newblock  \href{http://arxiv.org/abs/2207.05694}{arXiv:2207.05694}.
	
	\bibitem{RV14}
	Oded Regev and Thomas Vidick.
	\newblock ``Elementary proofs of {G}rothendieck theorems for completely bounded
	norms''.
	\newblock \href{https://dx.doi.org/10.7900/jot.2012jul02.1947}{J. Operat.
		Theor. {\bf 71}, 491--505}~(2014).
	
	\bibitem{Gro53}
	Alexander Grothendieck.
	\newblock ``R\'esum\'e de la th\'eorie m\'etrique des produits tensoriels
	topologiques''.
	\newblock Boll. Soc. Mat. S\~ao-Paulo {\bf 8}, 1--79~(1953).
	
	\bibitem{Pis12}
	Gilles Pisier.
	\newblock ``Grothendieck's theorem, past and present''.
	\newblock \href{https://dx.doi.org/10.1090/S0273-0979-2011-01348-9}{Bull. Amer.
		Math. Soc. {\bf 49}, 237--323}~(2012).
	
	\bibitem{LTW13}
	Debbie Leung, Ben Toner, and John Watrous.
	\newblock ``Coherent state exchange in multi-prover quantum interactive proof
	systems''.
	\newblock \href{https://dx.doi.org/10.4086/cjtcs.2013.011}{Chic. J. Theoret.
		Comput. Sci. {\bf 2013}, 11}~(2013).
	
	\bibitem{DSV15}
	Irit Dinur, David Steurer, and Thomas Vidick.
	\newblock ``A parallel repetition theorem for entangled projection games''.
	\newblock \href{https://dx.doi.org/10.1007/s00037-015-0098-3}{Comput. Complex.
		{\bf 24}, 201--254}~(2015).
	
	\bibitem{RV15}
	Oded Regev and Thomas Vidick.
	\newblock ``Quantum xor games''.
	\newblock \href{https://dx.doi.org/10.1145/2799560}{ACM Trans. Comput. Theory
		{\bf 7}, 15}~(2015).
	
	\bibitem{JLV20}
	Zhengfeng Ji, Debbie Leung, and Thomas Vidick.
	\newblock ``A three-player coherent state embezzlement game''.
	\newblock \href{https://dx.doi.org/10.22331/q-2020-10-26-349}{{Quantum} {\bf
			4}, 349}~(2020).
	
	\bibitem{MOA11}
	Albert~W. Marshall, Ingram Olkin, and Barry~C. Arnold.
	\newblock ``Inequalities: Theory of majorization and its applications''.
	\newblock \href{https://dx.doi.org/10.1007/978-0-387-68276-1}{Springer Series
		in Statistics}. Springer. ~(2011).
	
	\bibitem{NV01}
	Michael~A. Nielsen and Guifr\'{e} Vidal.
	\newblock ``Majorization and the interconversion of bipartite states''.
	\newblock \href{https://dx.doi.org/10.26421/QIC1.1-5}{Quantum Inf. Comput. {\bf
			1}, 76--93}~(2001).
	
	\bibitem{Nie02}
	Michael~A Nielsen.
	\newblock ``An introduction to majorization and its applications to quantum
	mechanics''.
	\newblock Lecture Notes, Department of Physics, University of Queensland,
	AustraliaPage~53~(2002).
	\newblock  url:~\url{https://michaelnielsen.org/papers/maj-book-notes.pdf}.
	
	\bibitem{Vid99}
	Guifr\'e Vidal.
	\newblock ``Entanglement of pure states for a single copy''.
	\newblock \href{https://dx.doi.org/10.1103/PhysRevLett.83.1046}{Phys. Rev.
		Lett. {\bf 83}, 1046--1049}~(1999).
	
	\bibitem{LP01}
	Hoi-Kwong Lo and Sandu Popescu.
	\newblock ``Concentrating entanglement by local actions: Beyond mean values''.
	\newblock \href{https://dx.doi.org/10.1103/PhysRevA.63.022301}{Phys. Rev. A
		{\bf 63}, 022301}~(2001).
	
	\bibitem{Hel69}
	Carl~W Helstrom.
	\newblock ``Quantum detection and estimation theory''.
	\newblock \href{https://dx.doi.org/10.1007/BF01007479}{J. Stat. Phys. {\bf 1},
		231--252}~(1969).
	
	\bibitem{Hol73}
	Alexander~S Holevo.
	\newblock ``Statistical decision theory for quantum systems''.
	\newblock \href{https://dx.doi.org/10.1016/0047-259X(73)90028-6}{J.
		Multivariate Anal. {\bf 3}, 337--394}~(1973).
	
	\bibitem{Wat18}
	John Watrous.
	\newblock ``Similarity and distance among states and channels''.
	\newblock \href{https://dx.doi.org/10.1017/9781316848142.004}{Pages 124--200}.
	\newblock Cambridge University Press. ~(2018).
	
	\bibitem{Tor70}
	Erik~Nikolai Torgersen.
	\newblock ``Comparison of experiments when the parameter space is finite''.
	\newblock \href{https://dx.doi.org/10.1007/BF00534598}{Probab. Theory Relat.
		Fields {\bf 16}, 219--249}~(1970).
	
	\bibitem{Tor91}
	Erik Torgersen.
	\newblock ``Comparison of statistical experiments''.
	\newblock \href{https://dx.doi.org/DOI: 10.1017/CBO9780511666353}{Encyclopedia
		of Mathematics and its Applications}. Cambridge University Press.
	Cambridge~(1991).
	
	\bibitem{Ren16}
	Joseph~M. Renes.
	\newblock ``Relative submajorization and its use in quantum resource
	theories''.
	\newblock \href{https://dx.doi.org/10.1063/1.4972295}{J. Math. Phys {\bf 57},
		122202}~(2016).
	
	\bibitem{HOS18}
	Micha{\l} Horodecki, Jonathan Oppenheim, and Carlo Sparaciari.
	\newblock ``Extremal distributions under approximate majorization''.
	\newblock \href{https://dx.doi.org/10.1088/1751-8121/aac87c}{J. Phys. A {\bf
			51}, 305301}~(2018).
	
	\bibitem{TFG23}
	Thomas Theurer, Kun Fang, and Gilad Gour.
	\newblock ``Single-shot entanglement manipulation of states and channels
	revisited''~(2023).
	\newblock  \href{http://arxiv.org/abs/2312.17088}{arXiv:2312.17088}.
	
	\bibitem{Gur03}
	Leonid Gurvits.
	\newblock ``Classical deterministic complexity of {E}dmonds' problem and
	quantum entanglement''.
	\newblock In Proceedings of the Thirty-Fifth Annual ACM Symposium on Theory of
	Computing.
	\newblock \href{https://dx.doi.org/10.1145/780542.780545}{Pages 10--19}.
	\newblock STOC '03New York, NY, USA~(2003). Association for Computing
	Machinery.
	
	\bibitem{NMCJW15}
	N~H~Y Ng, L~Man{\v{c}}inska, C~Cirstoiu, J~Eisert, and S~Wehner.
	\newblock ``Limits to catalysis in quantum thermodynamics''.
	\newblock \href{https://dx.doi.org/10.1088/1367-2630/17/8/085004}{New J. Phys.
		{\bf 17}, 085004}~(2015).
	
	\bibitem{GMNSYH15}
	Gilad Gour, Markus~P. M{\"u}ller, Varun Narasimhachar, Robert~W. Spekkens, and
	Nicole {Yunger Halpern}.
	\newblock ``The resource theory of informational nonequilibrium in
	thermodynamics''.
	\newblock
	\href{https://dx.doi.org/https://doi.org/10.1016/j.physrep.2015.04.003}{Phys.
		Rep. {\bf 583}, 1--58}~(2015).
	
	\bibitem{BCP14}
	T.~Baumgratz, M.~Cramer, and M.~B. Plenio.
	\newblock ``{{Quantifying Coherence}}''.
	\newblock \href{https://dx.doi.org/10.1103/PhysRevLett.113.140401}{Phys. Rev.
		Lett. {\bf 113}, 140401}~(2014).
	
	\bibitem{BDG15}
	Shuanping Du, Zhaofang Bai, and Yu~Guo.
	\newblock ``{{Conditions for coherence transformations under incoherent
			operations}}''.
	\newblock \href{https://dx.doi.org/10.1103/PhysRevA.91.052120}{Phys. Rev. A
		{\bf 91}, 052120}~(2015).
	
	\bibitem{CZZZ19}
	Senrui Chen, Xingjian Zhang, You Zhou, and Qi~Zhao.
	\newblock ``One-shot coherence distillation with catalysts''.
	\newblock \href{https://dx.doi.org/10.1103/PhysRevA.100.042323}{Phys. Rev. A
		{\bf 100}, 042323}~(2019).
	
	\bibitem{ST22}
	Ryuji Takagi and Naoto Shiraishi.
	\newblock ``Correlation in catalysts enables arbitrary manipulation of quantum
	coherence''.
	\newblock \href{https://dx.doi.org/10.1103/PhysRevLett.128.240501}{Phys. Rev.
		Lett. {\bf 128}, 240501}~(2022).
	
	\bibitem{BHORS13}
	Fernando G. S.~L. Brand\~ao, Micha\l{} Horodecki, Jonathan Oppenheim, Joseph~M.
	Renes, and Robert~W. Spekkens.
	\newblock ``{{Resource Theory of Quantum States Out of Thermal Equilibrium}}''.
	\newblock \href{https://dx.doi.org/10.1103/PhysRevLett.111.250404}{Phys. Rev.
		Lett. {\bf 111}, 250404}~(2013).
	
	\bibitem{HJ13}
	Micha{\l} Horodecki and Jonathan Oppenheim.
	\newblock ``Fundamental limitations for quantum and nanoscale thermodynamics''.
	\newblock \href{https://dx.doi.org/10.1038/ncomms3059}{Nat. Commun. {\bf 4},
		2059}~(2013).
	
	\bibitem{BHNOW15}
	Fernando Brand{\~a}o, Micha{\l} Horodecki, Nelly Ng, Jonathan Oppenheim, and
	Stephanie Wehner.
	\newblock ``The second laws of quantum thermodynamics''.
	\newblock \href{https://dx.doi.org/10.1073/pnas.1411728112}{Proc. Natl. Acad.
		Sci. U.S.A. {\bf 112}, 3275--3279}~(2015).
	
	\bibitem{LBS21}
	Patryk Lipka-Bartosik and Paul Skrzypczyk.
	\newblock ``All states are universal catalysts in quantum thermodynamics''.
	\newblock \href{https://dx.doi.org/10.1103/PhysRevX.11.011061}{Phys. Rev. X
		{\bf 11}, 011061}~(2021).
	
	\bibitem{RT22}
	Roberto Rubboli and Marco Tomamichel.
	\newblock ``Fundamental limits on correlated catalytic state transformations''.
	\newblock \href{https://dx.doi.org/10.1103/PhysRevLett.129.120506}{Phys. Rev.
		Lett. {\bf 129}, 120506}~(2022).
	
	\bibitem{CS15}
	Giulio Chiribella and Carlo~Maria Scandolo.
	\newblock ``Entanglement and thermodynamics in general probabilistic
	theories''.
	\newblock \href{https://dx.doi.org/10.1088/1367-2630/17/10/103027}{New J. Phys
		{\bf 17}, 103027}~(2015).
	
	\bibitem{RV73}
	A.~Wayne Roberts and Dale~E. Varberg.
	\newblock ``Convex functions''.
	\newblock Volume~57 of Pure and Applied Mathemathics; A Series of Monographs
	and Textbooks.
	\newblock ACADEMIC PRESS New York and London. ~(1973).
	\newblock
	url:~\url{https://www.sciencedirect.com/bookseries/pure-and-applied-mathematics/vol/57}.
	
	\bibitem{Neu37}
	John von Neumann.
	\newblock ``Some matrix-inequalities and metrization of matrix-space''.
	\newblock Tomsk. Univ. Rev. {\bf 1}, 286--300~(1937).
	
	\bibitem{Mir75}
	L.~Mirsky.
	\newblock ``A trace inequality of {J}ohn von {N}eumann''.
	\newblock \href{https://dx.doi.org/10.1007/BF01647331}{Monatsh. Math {\bf 79},
		303--306}~(1975).
	
	\bibitem{Fer67}
	Thomas~S Ferguson.
	\newblock ``Mathematical statistics: A decision theoretic approach''.
	\newblock \href{https://dx.doi.org/10.1016/C2013-0-07705-5}{Probability and
		Mathematical Statistics: A Series of Monographs and Textbooks}. Academic
	Press. ~(1967).
	
	\bibitem{Cor22}
	Thomas~H Cormen, Charles~E Leiserson, Ronald~L Rivest, and Clifford Stein.
	\newblock ``Introduction to algorithms''.
	\newblock MIT press. ~(2022).
	\newblock
	url:~\url{https://mitpress.mit.edu/9780262046305/introduction-to-algorithms/}.
	
	\bibitem{Wij85}
	Robert~A. Wijsman.
	\newblock ``A useful inequality on ratios of integrals, with application to
	maximum likelihood estimation''.
	\newblock \href{https://dx.doi.org/10.2307/2287917}{J. Am. Stat. Assoc {\bf
			80}, 472--475}~(1985).
	
	\bibitem{Rud87}
	W.~Rudin.
	\newblock ``Real and complex analysis''.
	\newblock Higher Mathematics Series. McGraw-Hill Education. ~(1987).
	\newblock
	url:~\url{https://www.mheducation.com/highered/product/M9780070542341.html?exactIsbn=true}.
	
\end{thebibliography}
\end{document}